\documentclass[12pt]{amsart}

\usepackage{amsmath,amsthm,amsfonts,amssymb,eucal, bbm}

\usepackage{color}

\renewcommand {\a}{ \alpha }
\renewcommand{\b}{\beta}

\newcommand{\g}{\gamma}
\newcommand{\G}{\Gamma}
\newcommand{\U}{\Upsilon}
\renewcommand{\k}{\kappa}

\renewcommand{\d}{\delta}

\renewcommand{\l}{\lambda}

\newcommand{\z}{\zeta}
\renewcommand{\t}{\theta}

\newcommand{\p}{\partial}
\newcommand{\om}{\omega}
\newcommand{\Om}{\Omega}

\newcommand{\oq}{\ {\raise 7pt\hbox{${\scriptstyle\circ}$}}
	\kern -7pt{
		\hbox{$Q$}}}

\newcommand{\R}{ \mathbb R}

\newcommand {\BS}{\mathbf S}

\newcommand {\BP}{\mathbf P}
\newcommand {\BQ}{\mathbf Q}

\newcommand {\bx}{\mathbf x}

\newcommand {\bk}{\mathbf k}

\newcommand {\bm}{\mathbf m}
\newcommand {\bl}{\mathbf l}
\newcommand {\bw}{\mathbf w}
\newcommand {\bz}{\mathbf z}

\newcommand {\bs}{\mathbf s}

\newcommand {\bj}{\mathbf j}

\newcommand{\SR}{{\sf{R}}}

\newcommand{\SD}{{\sf D}}

\newcommand{\SQ}{{\sf Q}}
\newcommand{\SP}{{\sf P}}
\newcommand{\SfS}{{\sf S}}

\newcommand{\CL}{\mathcal L}

\newcommand{\CH}{\mathcal H}

\newcommand{\CD}{\mathcal D}

\newcommand{\plainC}[1]{\textup{{\textsf{C}}}^{#1}}

\newcommand{\plainH}[1]{\textup{{\textsf{H}}}^{#1}}

\newcommand{\plainL}[1]{\textup{{\textsf{L}}}^{#1}}

\DeclareMathOperator{\card}{{card}}

\DeclareMathOperator {\re} {{Re}}

\DeclareMathOperator{\supp}{{supp}}

\newtheorem{thm}{Theorem}[section]
\newtheorem{cor}[thm]{Corollary}
\newtheorem{lem}[thm]{Lemma}
\newtheorem{prop}[thm]{Proposition}

\theoremstyle{definition}
\newtheorem{defn}[thm]{Definition}

\newtheorem{rem}[thm]{Remark}

\numberwithin{equation}{section}

\newcommand{\bee}{\begin{equation}}
	\newcommand{\ene}{\end{equation}}
\newcommand{\bees}{\begin{equation*}}
	\newcommand{\enes}{\end{equation*}}
\newcommand{\bes}{\begin{split}}
	\newcommand{\ens}{\end{split}}

\newcommand{\bet}{\begin{thm}}
	\newcommand{\ent}{\end{thm}}
\newcommand{\bel}{\begin{lem}}
	\newcommand{\enl}{\end{lem}}
\newcommand{\bec}{\begin{cor}}
	\newcommand{\enc}{\end{cor}}
\newcommand{\bep}{\begin{proof}}
	\newcommand{\enp}{\end{proof}}
\newcommand{\ber}{\begin{rem}}
	\newcommand{\enr}{\end{rem}}

\begin{document}
	\hoffset -4pc

\title
[One-particle density matrix]
{Analyticity of the one-particle density matrix}
\author{Peter Hearnshaw}
\author{Alexander V. Sobolev}
\address{Department of Mathematics\\ University College London\\
	Gower Street\\ London\\ WC1E 6BT UK}
\email{peter.hearnshaw.18@ucl.ac.uk}
\email{a.sobolev@ucl.ac.uk}

\begin{abstract}	 
It is proved that the one-particle density matrix $\g(x, y)$ for multi-particle systems 
is real analytic away from the nuclei and from the diagonal $x = y$. 
\end{abstract}


\maketitle

\section{Introduction} 

The objective of the paper is to study analytic properties of the one-particle density 
matrix for the molecule, consisting of $N$ electrons and $N_0$ nuclei 
described by the following Schr\"odinger operator:
\begin{align}\label{eq:ham}
\sum_{k=1}^N \bigg(-\Delta_k - \sum_{l=1}^{N_0}
\frac{Z_l}{|x_k - R_l|}
\bigg) 
 + \sum_{1\le j< k\le N} \frac{1}{|x_j-x_k|} + \sum_{1\le k < l\le N_0} \frac{Z_l Z_k}{|R_l-R_k|},
\end{align}
 where $R_l\in \R^3$ and $Z_l>0$, $l = 1, 2, \dots, N_0$, are the 
 positions and the charges, respectively, of $N_0$ nuclei, and    
$x_j\in\R^3$, $j = 1, 2, \dots, N$ are positions of $N$ electrons. The notation $\Delta_k$ is used for 
the Laplacian w.r.t. the variable $x_k$. 
The positions of the nuclei are assumed to be fixed, and as a result the 
very last term in \eqref{eq:ham} is constant. Thus 
in what follows we drop this term and instead of \eqref{eq:ham} we study the operator
\begin{align}\label{eq:oper}
H = H^{(0)} + V,\quad H^{(0)} =  -\Delta = -\sum_{k=1}^N\Delta_k
%
\end{align}  
with
\begin{align}\label{eq:coulomb}
V(\bx) = V^{\rm C}(\bx) = - \sum_{k=1}^N\sum_{l=1}^{N_0}
\frac{Z_l}{|x_k - R_l|}
 + \sum_{1\le j< k\le N} \frac{1}{|x_j-x_k|}.
\end{align}
This operator acts on the Hilbert space $\plainL2(\R^{3N})$ and it is self-adjoint on the domain 
$D(H) = D(H^{(0)}) = \plainH2(\R^{3N})$, since $V$ is infinitesimally 
$H^{(0)}$-bounded, see e.g. \cite[Theorem X.16]{ReedSimon2}. 

Let $\psi = \psi(\bx)$ 
%
%
be an eigenfunction of the operator $H$ with an eigenvalue $E\in\R$, i.e.
\begin{align}\label{eq:eigen}
(H-E)\psi = 0.
\end{align}  
For each $j=1, \dots, N$, we represent 
\begin{align*}
\bx = (x_j, \hat\bx_j), \quad \textup{where}\ 
\hat\bx_j = (x_1, \dots, x_{j-1}, x_{j+1},\dots, x_N),
\end{align*}
with obvious modifications if $j=1$ or $j=N$. 
The one-particle density matrix is defined as the function 
\begin{align}\label{eq:dem}
\tilde \g(x, y) = \sum_{j=1}^N\int\limits_{\R^{3N-3}} 
\psi(x, \hat\bx_j) \overline{\psi(y, \hat\bx_j)}\  
d\hat\bx_j,\quad (x,y)\in\R^3\times\R^3,
\end{align} 
see \cite{FHOS2004}. 
This function is one of the key objects in 
the  multi-particle quantum mechanics, see e.g. 
\cite{RDM2000}, \cite{LiebSei2010}, \cite{LLS2019}  for details and futher references. 
If one assumes that all $N$ particles are fermions 
(resp. bosons), i.e. that the function $\psi$ is 
antisymmetric (resp. symmetric) under the permutations $x_j\leftrightarrow x_k$, 
then the definition \eqref{eq:dem} simplifies:
\begin{align*}
\tilde \g(x, y) = N \int_{\R^{3N-3}} \psi(x, \hat\bx) \overline{\psi(x, \hat\bx)} d\hat\bx, 
\end{align*}
where we have denoted $\hat\bx = \hat\bx_1$, so $\bx = (x_1, \hat\bx)$. 
Our objective is to study the regularity properties of 
the function $\tilde\g(x, y)$ in the general form \eqref{eq:dem} without any symmetry assumptions. 

Regularity properties of solutions of elliptic equations is a classical and widely studied subject. 
For instance, 
it immediately follows from the general theory, see e.g. \cite{Hor1976}, that 
any local solution of \eqref{eq:eigen} is 
real analytic away from the singularities of the potential 
\eqref{eq:coulomb}. In his famous paper \cite{Kato1957} T. Kato showed that a local solution 
is locally Lipschitz with ``cusps" at the points of particle coalescence. 
Further regularity results include \cite{FHOS2005}, \cite{FHOS2009}, \cite{FS2018}. 
We cite the most recent paper \cite{FS2018} for further references. 

As far as the one-particle density matrix \eqref{eq:den} is concerned, 
in the analytic literature a special attention has been paid  to \textit{the one-particle density} 
$\tilde\rho(x) = \tilde\g(x, x)$. 
It was shown in \cite{FHOS2002}, that in spite of the nonsmoothness of $\psi$, 
the density $\tilde\rho(x)$ remains smooth as long as $x\not =R_l, l = 1, 2, \dots, N_0$, 
because of the averaging in $\hat\bx$.   
Moreover, the same authors prove in \cite{FHOS2004} that $\tilde\rho$ is in fact 
real analytic away from the nuclei, see also \cite{Jecko2010} for an alternative proof.  
Paper \cite{FHOS2004} also claims that the proofs therein imply the analyticity 
of $\tilde\g(x, y)$ for all $x, y$ away from the nuclei. However, 
the methods of \cite{FHOS2004} do not suffice to substantiate this claim. 
The objective of the current paper is to 
bridge this gap and 
prove the real analyticity for the one-particle density matrix 
$\tilde\g(x, y)$ for all $x\not = y$ away from the nuclei. We emphasize that the condition 
$x\not = y$ is not just an annoying technical restriction -- 
the function $\tilde\g(x, y)$ genuinely cannot be 
infinitely smooth on the diagonal. 
We have no analytic proof 
of this fact, but we can present an indirect justification. As shown in \cite{Sob2021}, 
the eigenvalues $\l_k(\G)$, $k = 1, 2, \dots$, 
of the non-negative compact operator 
$\G:\plainL2(\R^3)\to\plainL2(\R^3)$ 
with kernel $\tilde\g(x, y)$ decay at the rate of $k^{-8/3}$ as $ k\to\infty$. In general, 
for integral operators 
the rate of decay is dictated by the smoothness of the kernel: 
the smoother the kernel is, the faster 
the eigenvalues decrease, see bibliography in \cite{Sob2021}. If the function $\tilde\g(x, y)$ were 
infinitely differentiable for all $x$ and $y$ including the diagonal $x = y$ (but excluding the  nuclei), then the eigenvalues of $\G$  
would decay faster than any negative power of their number, which is not the case. 
To justify the non-smoothness for $x = y$ we also mention a heuristic argument presented in \cite{Cioslowski2020} that 
suggests that $\tilde\g(x, y)-\tilde\g(x, x)$ 
should behave as  $|x-y|^5$ for $x$ close to $y$. In fact, order $5$ of this homogeneous singularity 
is consistent with the $k^{-8/3}$-decay of the eigenvalues, see again 
\cite{Sob2021} and bibliography therein. 

One should also say that the real analyticity of the density 
$\tilde\rho(x) = \tilde\g(x, x)$, $x\not = 0$, established in \cite{FHOS2004}, 
means that the density matrix $\tilde\g(x, y)$ is analytic \textit{along} 
the diagonal $x = y$ as a function of one variable $x$, 
which does not contradict the non-smoothness of $\tilde\g(x, y)$ \textit{on} the diagonal as a function 
of the two variables $x$ and $y$. 

Before we state the main result note that regularity 
of each of the terms in \eqref{eq:dem} 
can be studied individually. Furthermore, it suffices to establish the real 
analyticity of the function 
\begin{align}\label{eq:den}
\g(x, y) = \int\limits_{\R^{3N-3}} \psi(x, \hat\bx) \overline{\psi(y, \hat\bx)}\  
d\hat\bx,\quad (x,y)\in\R^3\times\R^3.
\end{align} 
The remaining terms in the sum \eqref{eq:dem} 
are handled by relabeling the variables. Thus from now on we  
use the terms 
\textit{one-particle density matrix} and  
\textit{one-particle density} for the functions \eqref{eq:den} and $\rho(x) = \g(x, x)$ respectively.

The next theorem constitutes our main result. 

\begin{thm}\label{thm:main}
Let the function $\g(x, y)$, $(x, y)\in \R^3\times\R^3$, be defined by \eqref{eq:den}. 
Then $\g(x, y)$ is real analytic as a function of variables $x$ and $y$ on the set 
\begin{align}\label{eq:D0}
\CD_0 = \{(x, y)\in \R^3\times\R^3: x\not = R_l, 
y \not = R_l, l = 1, 2, \dots N_0,\  {\rm and} \ x\not=y\}.
\end{align}
\end{thm}

As mentioned above, the eigenfunction $\psi(\bx)$ 
loses smoothness at the particle coalescence points. 
Therefore, a direct differentiation of 
\eqref{eq:den} w.r.t. $x$ and $y$ under the integral will not produce the required analyticity. 
In order to circumvent this difficulty we use the property 
that the eigenfunction preserves smoothness even at the coalescence points if 
one replaces the standard derivatives by cleverly chosen directional ones. 
For example, the function $\psi$ is infinitely smooth in the variable $x_1+x_2+\dots x_N$,  
as long as  
$x_j\not = R_l$ for each $j = 1, 2, \dots, N$, $l = 1, 2, \dots, N_0$. In other words, it is infinitely differentiable 
with respect to the directional derivative 
\begin{align*}
\SD = \sum_{j=1}^N\nabla_{x_j}. 
\end{align*}
Such regularity follows from the fact that the potential 
\eqref{eq:coulomb} is smooth w.r.t. $\SD$ on the same domain. In particular, 
\begin{align*}
\SD\, \frac{1}{|x_j-x_k|} = 0,\quad \textup{for all}\quad j\not = k.
\end{align*} 
This approach was successfully used in \cite{FHOS2004} (or even in the 
earlier paper \cite{FHOS2002}) 
in the study of the electron density $\rho(x)=\g(x, x)$. 
To illustrate the use of the directional derivatives
in the study of $\rho(x)$, below we give a simplified example. 

Assume that $N_0=1$ and that $R_1 = 0$.  For further simplification 
instead of the function $\g(x, x)$ we consider only a part of it. Precisely, 
let $\z\in\plainC\infty_0(\R^3)$ be a real-valued function such that $\z(t) = 0$ for $|t|>\varepsilon/2$ 
with some $\varepsilon>0$.  Let us show that the integral 
\begin{align}\label{eq:F}
F(x) = \int_{\R^{3N-3}} |\psi(x, \hat\bx)|^2 \prod_{j=2}^N \z(x-x_j)\, d\hat\bx
\end{align}
is $\plainC\infty$ for all $x\in\R^3$ such that $|x|>\varepsilon$. 
The presence of the cut-off functions in \eqref{eq:F} means that all the particles are within a $\varepsilon/2$ distance 
from the first particle. Thus on the domain of integration all variables are separated from $0$:
$|x_j|\ge |x| - |x-x_j|>\varepsilon/2$, $j = 2, 3, \dots, N$. 
Rewrite $F$ making the change of variables $x_j = w_j + x$, 
$j = 2, 3, \dots, N$, under the integral:
\begin{align*}
F(x) = \int_{\R^{3N-3}} |\psi(x, w_2+x, \dots, w_N+x)|^2
\prod_{j=2}^N \z(-w_j)\, d\hat\bw,\ \ \hat\bw = (w_2, w_3, \dots, w_N).
\end{align*}
Differentiating the integral w.r.t. $x$ we get:
\begin{align*}
\nabla_x F(x) = &\ 2\re \int_{\R^{3N-3}} (\SD \psi)(x, w_2+x, \dots, w_N+x)\\
&\ \qquad\qquad \times \overline{\psi(x, w_2+x, \dots, w_N+x)}
\prod_{j=2}^N \z(-w_j)\, d\hat\bw\\
= &\ 2\re \int_{\R^{3N-3}} (\SD \psi)(x, \hat\bx)\overline{\psi(x, \hat\bx)} 
\prod_{j=2}^N \z(x-x_j)\, d\hat\bx.
\end{align*}
Now it is clear that by virtue of smoothness of $\psi$ w.r.t. the derivative $\SD$, 
this relation can be differentiated arbitrarily many times, thereby proving that $F\in\plainC\infty$ 
for all $x: |x| >\varepsilon$. 

The complete proof of real analyticity of $\rho(x)$ in \cite{FHOS2004} 
is more involved. In particular, it requires the 
study of various cut-off functions that keep some of the particles 
``close" to each other, but separate from the rest 
of them (we call this group of particles \textit{the cluster associated 
with the given cut-off}). 
Adaptation of the above argument to such cases leads to the introduction of 
the \textit{cluster} derivatives 
(i.e. directional derivatives involving only the particles in a cluster), and  
it is far from straightforward.  
 
As in \cite{FHOS2004}, 
in the current paper our 
starting point is again a careful analysis of the smoothness properties of 
the eigenfunction $\psi$ with respect to the cluster derivatives. 
However we find the argument in \cite{FHOS2004} somewhat condensed and sketchy in places. 
Thus we provide our own proofs that contain more detail and at the same time, 
as we believe, are sometimes simpler than in \cite{FHOS2004}. 

To obtain bounds for the derivatives of $\g(x, y)$ we use again cluster derivatives, but the method 
of \cite{FHOS2004} is ineffective if applied directly. 
It has to be reworked taking into account the presence of two variables (i.e. $x, y$) instead of one. 
At the heart of our approach is the concept of an \textit{extended cut-off function} that depends 
on the variables $(x, y)\in\R^3\times\R^3$ and $\hat\bx\in \R^{3N-3}$. Any such function 
$\Phi(x, y, \hat\bx)$ has two clusters associated with it,  
whose properties are linked to each other (see Subsect. \ref{subsect:extended}).  
This enables us to apply the cluster derivatives method to integrals of the form 
\begin{align}\label{eq:pau}
\int\limits_{\R^{3N-3}} \psi(x, \hat\bx) \overline{\psi(y, \hat\bx)}
\Phi(x, y, \hat\bx)
\,  d\hat\bx,\quad (x,y)\in\R^3\times\R^3.
\end{align}
At the last stage we construct a partition of unity on $\R^{3N+3}$ which consists of 
extended cut-offs, and split $\g(x, y)$ in the sum of terms of the form \eqref{eq:pau}. Estimating each of them individually, we get the desired real analyticity.

The paper is organized as follows. 
In Sect. \ref{sect:main} we state Theorem 
\ref{thm:main1}, involving more general interactions 
between particles, that implies Theorem \ref{thm:main} as a special case. This step allows to include 
other physically meaningful potentials, such as, for example, the Yukawa potential.  
An important conclusion of this Section is that the claimed analyticity 
of the function $\g(x, y)$ follows from appropriate $\plainL2$-bounds on the derivatives of $\g(x, y)$, 
enunciated in Theorem \ref{thm:deriv2}. The rest of the paper is devoted to the proof of 
Theorem \ref{thm:deriv2}.   

Sect. \ref{sect:reg} is concerned with the study of the directional derivatives of the 
eigenfunction $\psi$. The main objective is 
to establish suitable $\plainL2$-estimates for higher order derivatives of $\psi$ on the open sets, separating 
different clusters of variables. Here our argument follows that 
of \cite{FHOS2004} with some simplifications. 
In Sect. \ref{sect:cutoff} we study in detail properties of smooth 
cut-off functions including 
the extended cut-offs $\Phi = \Phi(x, y, \hat\bx)$, $x, y\in\R^3, \hat\bx\in\R^{3N-3}$, 
and clusters associated with them. In Sect. \ref{sect:estimates} we put 
together the results of Sect. \ref{sect:reg} and \ref{sect:cutoff} to estimate the 
derivatives of integrals of the form \eqref{eq:pau} 
with extended cut-offs $\Phi$. These estimates 
are used to prove Theorem \ref{thm:deriv2} with the help 
of a partition of unity that consists of extended cut-offs. 
This completes the proof of 
  Theorem \ref{thm:main1}, and hence that of the main result, Theorem \ref{thm:main}.
The Appendix contains some elementary combinatorial formulas that are used throughout the proof.

We conclude the introduction with some general notational conventions.

\textit{Constants.}
By $C$ or $c$ with or without indices, we denote 
various positive constants whose exact value is of no importance.

\textit{Coordinates.} 
As mentioned earlier, we use the following standard notation for the coordinates: 
$\bx = (x_1, x_2, \dots, x_N)$,\ where $x_j\in \R^3$, $j = 1, 2, \dots, N$. 
Very often it is convenient to represent $\bx$ in the form 
$\bx = (x_1, \hat\bx)$ with  
$\hat\bx = (x_2, x_3, \dots, x_N)\in\R^{3N-3}$. 

\textit{Clusters.}
Let $\SR = \{1, 2, \dots, N\}$. 
An index set $\SP\subset\SR$ is called 
\textit{cluster}. 
The cluster $\SR$ is called \textit{maximal}. 
We denote $|\SP| = \card \SP$, 
$\SP^{\rm c} = \SR\setminus \SP$,\ $\SP^* = \SP\setminus\{1\}$. 
If $\SP = \varnothing$, then $|\SP| = 0$ and $\SP^{\rm c} = \SR$.

For $M$ clusters $\SP_1, \dots, \SP_M$ we write $\BP = \{\SP_1, \SP_2, \dots, \SP_M\}$, 
 $\BP^* = \{\SP_1^*, \SP_2^*, \dots, \SP_M^*\}$
and call $\BP$, $\BP^*$ \textit{cluster sets}. 
Clusters $\SP_1, \SP_2, \dots, \SP_M$ in a cluster set are not assumed to be all distinct. 

\textit{Derivatives.} 
Let $\mathbb N_0 = \mathbb N\cup\{0\}$.
If $x = (x', x'', x''')\in \R^3$ and $m = (m', m'', m''')\in \mathbb N_0^3$, then 
the derivative $\p_x^m$ is defined in the standard way:
\begin{align*}
\p_x^m = \p_{x'}^{m'}\p_{x''}^{m''}\p_{x'''}^{m'''}.
\end{align*}
This notation extends to $x\in\R^d$ with an arbitrary dimension $d\ge 1$ in the obvious way. 
Denote also 
\begin{align*}
\p^{\bm} =  
\p_{x_1}^{m_1} \p_{x_2}^{m_2}\cdots \p_{x_N}^{m_N},\quad 
\bm = (m_1, m_2, \dots, m_N)\in \mathbb N_0^{3N}.
\end{align*}
A central role is played by the following directional derivatives. 
For a cluster $\SP$ and each $m = (m', m'', m''')\in \mathbb N_0^3$, we define 
the \textit{cluster derivatives} 
\begin{align}\label{eq:clusterder}
\SD_{\SP}^m  = &\ \biggl(\sum_{k\in \SP} \p_{x_k'}\biggr)^{m'}
\bigg(\sum_{k\in \SP} \p_{x_k''}\bigg)^{m''}
\bigg(\sum_{k\in \SP} \p_{x_k'''} \bigg)^{m'''}.
\end{align}
These operations can be viewed as partial derivatives w.r.t. the variable 
$\sum_{k\in\SP} x_k$. 
Let $\BP = \{\SP_1, \SP_2, \dots, \SP_M\}$ be a cluster set, and let 
$\bm = (m_1, m_2, \dots, m_M)$, $m_k\in \mathbb N_0^3$, $k = 1, 2, \dots, M$. Then we denote 
\begin{align*} 
\SD_{\BP}^{\bm} = \SD_{\SP_1}^{m_1} \SD_{\SP_2}^{m_2}\cdots\SD_{\SP_M}^{m_M}.
\end{align*}  

\textit{Supports.} 
For any smooth function $f = f(\bx)$,  
we define $\supp f = \{\bx: f(\bx)\not = 0\}$. With this definition we immediately get the useful 
property that
\begin{align}\label{eq:supp}
\supp (fg) = \supp f\cap \supp g.
\end{align}
Furthermore, for any $\bm\in\mathbb N_0^{3N}$, $|\bm|=1$,  we have 
\begin{align}\label{eq:suppderiv}
\supp\p^{\bm} f\subset\supp f,\quad \textup{if}\quad f\ge 0.
\end{align}
 
\section{The main result}\label{sect:main}
 
\subsection{Main theorem}
The main theorem \ref{thm:main} 
is derived from the following result, 
that holds for more general potentials than \eqref{eq:coulomb}.

Let $V_{k,l}, W_{k,j}\in \plainC\infty(\R^3\setminus \{0\})$, $l = 1, 2, \dots, N_0$, 
$k, j = 1, 2, \dots, N$, 
be functions on $\R^3$ such that 
for all $v\in \plainH1(\R^{3})$ we have 
\begin{align}\label{eq:form}
\| V_{k, l} v\|_{\plainL2} + \|W_{k, j} v\|_{\plainL2} \le C\| v\|_{\plainH1},
\end{align}
and for every $\varepsilon>0$, we have 
\begin{align}\label{eq:potan}
\sum_{k=1}^N \sum_{l=1}^{N_0} \max_{|x|>\varepsilon}|\p^m_x V_{k, l}(x)| 
+ \sum_{\substack{k, j = 1\\k\not = j}}^{N}\max_{|x|>\varepsilon}|\p^m_x W_{k, j}(x)|
\le A_0^{1+|m|} (1+|m|)^{|m|},
\end{align}
for all $l = 1, 2, \dots, N_0,\ k, j = 1, 2, \dots, N$ with some positive 
constant $A_0 = A_0(\varepsilon)$. The condition 
\eqref{eq:potan}  
implies that the functions $V_{k, l}$ and $W_{k, j}$ are real analytic on $\R^3\setminus\{0\}$.  
Instead of the potential $V^{\rm C}$ defined in \eqref{eq:coulomb}, we 
consider the potential  
\begin{align}\label{eq:genpot}
V(\bx) = \sum_{k=1}^N \sum_{l=1}^{N_0} V_{k, l}(x_k-R_l) 
+ \sum_{\substack{k, j = 1\\k\not = j}}^{N} W_{k, j}(x_k-x_j).
\end{align}  
The Coulomb potentials $V_{k, l} (x) = -Z_l|x|^{-1}$ 
and $W_{k, j}(x) = (2|x|)^{-1}$ satisfy \eqref{eq:form} 
in view of the classical Hardy's inequality, see e.g. 
\cite[The Uncertainty Principle Lemma, p. 169]{ReedSimon2}. 
Furthermore, the bounds \eqref{eq:potan} can be 
deduced from the estimates for harmonic functions, established, e.g. in \cite[Theorem 7, p. 29]{Evans1998}.
Thus the potential 	\eqref{eq:coulomb} is a special case of \eqref{eq:genpot}. 
Working with more general potentials allows one to include into consideration other physically meaningful 
interactions, such as, e.g., the Yukawa potential. 
This generalization was pointed out in \cite{FHOS2004}.

We need the following elementary elliptic regularity fact, which we 
give with a proof, since it is quite short.   

\begin{lem}\label{lem:reg}
Suppose that $V$ is given by \eqref{eq:genpot}.  
 Then 
\begin{align}\label{eq:hardy}
\|V v\|_{\plainL2}\le C\|v\|_{\plainH1},\ 
\end{align}				                                                              
for all $v\in \plainH1(\R^{3N})$. 

If $v\in \plainH1(\R^{3N})$ and $H v\in\plainL2(\R^{3N})$, then 
$v\in \plainH2(\R^{3N})$ and  
\begin{align}\label{eq:domain}
\|v\|_{\plainH2}\le C\big(\|H v\|_{\plainL2} + \|v\|_{\plainL2}\big).
\end{align} 
The constant $C$ depends on $N$ and $N_0$ only.
 \end{lem}

\begin{proof}
The bound \eqref{eq:hardy} immediately follows from \eqref{eq:form}. 

For $v\in \plainH1$, $H v\in\plainL2$, it follows from \eqref{eq:hardy} that  
\begin{align}\label{eq:per}
-\Delta v = H v - V v \in \plainL2. 
\end{align}
Consequently, in view of the straightforward bound
\begin{align}\label{eq:delta}
\|v\|_{\plainH2}\le C_1\big(\|\Delta v\|_{\plainL2} + \|v\|_{\plainL2}\big), 
\end{align}
the function $v$ is $\plainH2$, and hence \eqref{eq:hardy} implies that  
\begin{align}\label{eq:inf}
\|V v\|_{\plainL2}
\le \d \|v\|_{\plainH2} + \tilde C_\d \|v\|_{\plainL2},
\end{align}
for all $\d > 0$. Together with \eqref{eq:per} and \eqref{eq:delta} this 
leads to the bound 
\begin{align*}
\|v\|_{\plainH2}\le C_1\big(\|H v\|_{\plainL2} + \d\|v\|_{\plainH2} 
+ (\tilde C_\d + 1)\|v\|_{\plainL2}\big). 
\end{align*}
Taking $\d = (2 C_1)^{-1}$, we easily derive \eqref{eq:domain} with a suitable constant $C>0$. 
\end{proof}
 
Note that the estimate \eqref{eq:inf} shows that the potential 
\eqref{eq:genpot} 
is infinitesimally $H^{(0)}$-bounded, so that the operator  $H$ defined in 
\eqref{eq:oper} is self-adjoint on 
the domain $D(H^{(0)}) = \plainH2(\R^{3N})$.

Theorem \ref{thm:main} is a consequence of the following result. 

\begin{thm}\label{thm:main1} Let the potential $V$ be given by \eqref{eq:genpot}, with 
some functions $V_{k, l}$ and $W_{k, j}$, satisfying the conditions \eqref{eq:form} and \eqref{eq:potan}. 
Let $\psi$ be an eigenfunction of the operator \eqref{eq:oper}, and let 
$\g(x, y)$ be as defined in \eqref{eq:den}. 
Then $\g(x, y)$ is real analytic as a function of variables $x$ and $y$ on the set \eqref{eq:D0}.
\end{thm}

For the sake of simplicity we prove this theorem only for the case of a single atom, 
i.e. for $N_0=1$. 
The general case requires only obvious modifications. 
Without loss of generality we set $R_1 = 0$. 
Thus \eqref{eq:genpot} rewrites as 
\begin{align}\label{eq:genpot1}
V(\bx) = \sum_{k=1}^N V_{k}(x_k) 
+ \sum_{\substack{k, j = 1\\k\not = j}}^{N} W_{k, j}(x_k-x_j),\ V_k = V_{k, 1},
\end{align}
and the stated analyticity of $\g(x, y)$ will be proved on the set
\begin{align*}
\CD_0 = \{(x, y)\in \R^3\times \R^3: x\not = 0, y\not=0, x\not=y\}.
\end{align*}
This result is derived from the following $\plainL2$-bound on the set
\begin{align}\label{eq:deps}
\CD = \CD_\varepsilon = 
\{(x, y)\in \R^3\times \R^3: |x|>\varepsilon, |y|> \varepsilon,\ |x-y|>\varepsilon\},\quad \varepsilon>0.
\end{align}

\begin{thm} \label{thm:deriv2}
Let $\varepsilon>0$ be arbitrary. 
 Then, for all $k, m\in \mathbb N_0^3$, we have 
\begin{align}\label{eq:der}
\|\p_{x}^{k}\p_y^{m}\gamma(\ \cdot\ , \ \cdot\ )\|_{\plainL2(\CD_\varepsilon)}
\le A^{|k|+|m|+2} (|k|+|m|+1)^{|k|+|m|},
\end{align}
with some constants $A = A(\varepsilon)$, independent of $k, m$.
\end{thm}

The derivation of Theorem \ref{thm:main1} from Theorem \ref{thm:deriv2} 
is based on the following elementary lemma. 
 
\begin{lem}\label{lem:2toinfty}
Let $\Om\subset\R^d$ be an open set, and let 
Let $f\in\plainC\infty(\Om)$ be a function such that 
\begin{align}\label{eq:deriv2}
\|\p_x^s f\|_{\plainL2(\Om)}\le B^{2+|s|} (1+|s|)^{|s|},
\end{align} 
for all $s\in \mathbb N_0^d$, with some positive constant $B$. Then $f$ is real analytic on $\Om$. 
\end{lem} 
 
\begin{proof}
Let $x_0\in\Om$, and let $r >0$ be such that $B(x_0, 2r)\subset \Om$. We aim to prove that 
\begin{align}\label{eq:factorial}
|\p_x^s f(x)|\le C R^{-|s|}s!, \quad \forall s\in\mathbb N_0^d,
\end{align}
for each $x\in B(x_0, r)$, with some  positive constants $C$ and $R$, possibly depending on $x_0$. 
According to \cite[Proposition 2.2.10]{Krantz_Parks} this would 
imply the required analyticity.

Let $\b\in\plainC\infty_0(\R^d)$ be a function supported on $B(x_0, 2r)$ and such that 
$\b = 1$ on $B(x_0, r)$. Denote 
\begin{align*}
g(x) = \b(x) \p_x^s f(x).
\end{align*}    
For $l > d/4$ we can estimate 
\begin{align*}
\|g\|_{\plainL\infty(\R^d)}
\le C\|(1-\Delta)^l g\|_{\plainL2(\R^d)},
\end{align*}
with a constant $C$ depending on $l$. 
Now it follows from \eqref{eq:deriv2} that 
\begin{align*}
\|g\|_{\plainL\infty(\R^d)}
\le C' B^{|s|+2l+2}(|s|+2|l|+1)^{|s|+2l}.
\end{align*}
By \eqref{eq:comb1}, the right-hand side does not exceed 
\begin{align*}
\tilde C (B e)^{|s|+2l+2} (|s| + 2l)!
\le 
\tilde C (B e)^{|s|+2l+2} e^{2l(|s|+2l)} |s|!
\end{align*}
According to \eqref{eq:spr}, 
\begin{align*}
|s|!\le d^{|s|}s!.
\end{align*}
Consequently,
\begin{align*}
\|g\|_{\plainL\infty(\R^d)}
\le \tilde C (Be)^{2l+2} e^{4l^2} 
(Be^{1+2l} d)^{|s|} s!.
\end{align*}
This bound leads to \eqref{eq:factorial} with explicitly given constants $C$ and 
$R$. The proof is now complete.
\end{proof} 
 
\begin{proof}[Proof of Theorem \ref{thm:main1}] 
According to Theorem \ref{thm:deriv2} and Lemma 
\ref{lem:2toinfty}, the function $\g(x, y)$ is real analytic on $\CD_\varepsilon$ for all $\varepsilon>0$. 
Consequently, it is real analytic on 
\begin{align*}
\CD_0 = \bigcup_{\varepsilon>0} \CD_\varepsilon, 
\end{align*}
as required. 
\end{proof}

The rest of the paper is focused on the proof of Theorem \ref{thm:deriv2}.

\subsection{More notation}
Here we introduce some important sets in $\R^{3N}$ and $\R^{3N-3}$. 
For $\varepsilon\ge 0$ introduce 
\begin{align}\label{eq:xp}
X_{\SP}(\varepsilon) = 
\begin{cases}
\R^{3N}\quad {\rm for}\ |\SP| = 0 \ \textup{or} \ N,\\[0.2cm]
\{\bx\in \R^{3N}: \ |x_j-x_k| > \varepsilon, 
\forall j\in \SP, k\in \SP^{\rm c} \},\quad  {\rm for}\ 0< |\SP| < N,
\end{cases}
\end{align} 
The set $X_\SP(\varepsilon), \varepsilon>0,$ separates the points $x_k$ 
and $x_j$ labeled by the clusters $\SP$ and $\SP^{\rm c}$ respectively. 
Note that $X_{\SP}(\varepsilon) = X_{\SP^{\rm c}}(\varepsilon)$. 

Define also the sets separating $x_k$'s from the origin:
\begin{align}\label{eq:tp}
T_\SP(\varepsilon) = 
\begin{cases}
\R^{3N}, \quad {\rm for}\ |\SP|=0,\\[0.2cm]
\{\bx\in\R^{3N}: |x_j|>\varepsilon,\ \forall j\in \SP\},\ \quad {\rm for}\ |\SP|>0.
\end{cases}
\end{align}
It is also convenient to introduce a similar notation involving only the variable $\hat\bx$:
\begin{align}\label{eq:hattp}
\widehat T_{\SP^*}(\varepsilon) = 
\begin{cases}
\R^{3N-3}, \quad {\rm for}\ |\SP^*|=0,\\[0.2cm]
\{\hat\bx\in\R^{3N-3}: |x_j|>\varepsilon,\ \forall j\in \SP^*\},\ \quad {\rm for}\ |\SP^*|>0.
\end{cases}
\end{align}
For the cluster sets 
$\BP = \{\SP_1, \SP_2, \dots, \SP_M\}$, 
 $\BP^* = \{\SP_1^*, \SP_2^*, \dots, \SP_M^*\}$ define   
\begin{align}\label{eq:up}
\begin{cases}
X_{\BP}(\varepsilon) = &\ \bigcap_{s=1}^M X_{\SP_s}(\varepsilon)\subset \R^{3N},
\quad T_{\BP}(\varepsilon) =  \bigcap_{s=1}^M T_{\SP_s}(\varepsilon)\subset \R^{3N},\\[0.3cm]
&\ 
U_{\BP}(\varepsilon) = X_{\BP}(\varepsilon)\cap T_{\BP}(\varepsilon),\\[0.3cm]
&\ \widehat T_{\BP^*}(\varepsilon) =  \bigcap_{s=1}^M \widehat T_{\SP_s^*}(\varepsilon)\subset\R^{3N-3}.
\end{cases}
\end{align} 

Now we introduce the standard cut-off functions with which we work. Let 
\begin{align}\label{eq:xi}
\xi\in \plainC\infty(\R):\ 0\le \xi(t)\le 1,\ 
\xi(t) = 
\begin{cases}
0,\ {\rm if}\ t \le 0,\\[0.2cm]
1,\ {\rm if}\ t\ge 1.
\end{cases}
\end{align}
Now we define two  radially-symmetric functions $\z, \t\in\plainC\infty(\R^3)$ as follows:
\begin{align}\label{eq:pu}
\t(x) = \xi\bigg(\frac{4N}{\varepsilon}|x|-1\bigg),\quad \z(x) = 1-\t(x), \quad x\in\R^3,
\end{align}
so that 
%
\begin{align*}
\t(x) = 0 \quad \textup{for}\quad x\in B\big(0, \varepsilon(4N)^{-1}\big),\qquad\qquad
\z(x) = 0 \quad \textup{for}\quad x\notin B\big(0, \varepsilon(2N)^{-1}\big).
%
\end{align*}

\section{Regularity of the eigenfunctions}\label{sect:reg}
   
In this section we establish estimates for the derivatives $\SD^\bm_\BP\psi$ 
of the eigenfunction $\psi$, see \eqref{eq:clusterder} for the 
definition of the cluster derivatives.   
Our argument is an expanded version of the approach suggested in \cite{FHOS2004}, which, in turn, was
inspired by the proof of analyticity for solutions of elliptic equations with analytic coefficients,
see e.g. the classical monograph \cite[Section 7.5]{Hor1976}. 

The key point of our argument is the regularity of the functions $\SD^\bm_\BP\psi$ for all 
$\bm\in\mathbb N_0^{3M}$ on the domain $U_\BP(\varepsilon)$ with arbitrary positive $\varepsilon$. 
As before, in the estimates below we denote by $C, c$ with or without indices positive constants 
whose exact value is of no importance. 
For constants that are important for subsequent results, we use the notation $L$ or $A$ with indices. 
The letter $L$ (resp. $A$) is used when the constant is independent of (resp. dependent on) $\varepsilon$.

We begin the proof of the required property with studying the regularity of 
the potential \eqref{eq:genpot1}. 

\subsection{Regularity of the potential \eqref{eq:genpot1}}
The next assertion is a more detailed variant of \cite[Part 2 of Lemma A.3]{FHOS2004} 
adapted to the potential \eqref{eq:genpot1}.

\begin{lem}\label{lem:V}
Let $V$ be as defined in \eqref{eq:genpot1} with $N_0 = 1$, 
and let $\BP=\{\SP_1, \SP_2, \dots, \SP_M\}$ 
be an arbitrary collection of clusters. 
Then for all $\bm\in\mathbb N_0^{3M}, |\bm|\ge 1$, the function 
$\SD^{\bm}_\BP V$ is $\plainC\infty$ 
on $U_\BP(\varepsilon)$, and the bound 
\begin{align}\label{eq:V}
\|\SD^{\bm}_\BP V\|_{\plainL\infty(U_\BP(\varepsilon))}
\le A_0^{1+|\bm|} (|\bm|+1)^{|\bm|}
\end{align}
holds, where 
$A_0$ is the constant from the condition \eqref{eq:potan}.
\end{lem}

\begin{proof} 
Without loss of generality we may assume that 
$\bm = (m_1, m_2, \dots, m_M)$ with all $|m_j|\ge 1$. 
Indeed, suppose that 
$m_1 = 0$ and represent $\BP = \{\SP_1, \tilde \BP\}$ with $\tilde\BP = \{\SP_2, \SP_3, \dots, \SP_M\}$. 
Then, denoting $\tilde\bm = (m_2, m_3, \dots, m_M)$, we get  
\begin{align*}
\|\SD^{\bm}_\BP V\|_{\plainL\infty(U_\BP(\varepsilon))}
= \|\SD^{\tilde\bm}_{\tilde\BP} V\|_{\plainL\infty(U_\BP(\varepsilon))}
\le \|\SD^{\tilde\bm}_{\tilde\BP} V\|_{\plainL\infty(U_{\tilde\BP}(\varepsilon))}. 
\end{align*}
Repeating, if necessary, this procedure we can eliminate all zero components of $\bm$, 
and the clusters, attached to them. 
Thus we assume henceforth that $|m_j|\ge 1$, $j = 1, 2, \dots, M$.

If $|m|=1$, then a direct differentiation gives the formula
\begin{align*}
\SD^m_{\SP_s} V_{k}(x_k)= 
\begin{cases}
0,\quad k\notin \SP_s,\\
\left.\p_x^m V_k(x)\right|_{x = x_k},\quad k\in \SP_s.
\end{cases}
\end{align*}
This function is $\plainC\infty$ on $U_{\SP_s}(\varepsilon)$, and 
further differentiation gives the same formula for all 
$|m|\ge 1$.
Similarly, 
\begin{align*}
\SD^m_{\SP_s} W_{kj}(x_k-x_j) = 
\begin{cases}
0,\quad k, j\in \SP_s\quad  {\rm or}\quad k, j\notin \SP_s,\\
\left.\p_x^m W_{kj}(x)\right|_{x = x_k-x_j},\quad k\in \SP_s, j\notin \SP_s.
\end{cases}
\end{align*}
Consequently, 
\begin{align*}
\SD_{\BP}^{\bm} V_k(x_k) = 
\begin{cases}
0,\quad k\notin \cap_s \SP_s,\\
\left.\p_x^{m_1+m_2+\dots+m_M} V_k(x)\right|_{x = x_k},\quad k\in \cap_s \SP_s.
\end{cases}
\end{align*}
and
\begin{align*}
\SD^{\bm}_{\BP} W_{kj}(x_k-x_j) = 
\begin{cases}
0,\quad j, k\in \cap_s \SP_s\quad  {\rm or}\quad k, j\notin \cap_s \SP_s,\\
\left.\p_x^{m_1+m_2+\dots+ m_M} W_{kj}(x)\right|_{x = x_k-x_j},
\quad k\in \cap_s \SP_s, j\notin \cap_s \SP_s.
\end{cases}
\end{align*}
These functions are $\plainC\infty$ on $U_{\BP}(\varepsilon)$, and, by the definition \eqref{eq:genpot1},  
it follows from \eqref{eq:potan}
that 
\begin{align*}
\|\SD_{\BP}^{\bm} V\|_{\plainL\infty(U_{\BP}(\varepsilon))}\le 
A_0^{1+|\bm|} (1+|\bm|)^{|\bm|}.
\end{align*}
This bound coincides with \eqref{eq:V}.
\end{proof}

Now we proceed to the study of the derivatives $\SD_\BP^{\bm}\psi$. 

\subsection{Regularity of the derivatives $\SD^{\bm}_\BP\psi$} 
As before, let $\psi\in\plainH2(\R^{3N})$ be an eigenfunction of 
the operator $H$, with the eigenvalue $E\in\R$, i.e. $H_E\psi = 0$, where $H_E = H-E$. 
Let $\BP = \{\SP_1, \SP_2, \dots, \SP_M\}$ be a cluster set, and 
consider the function $u_\bm =\SD_{\BP}^{\bm}\psi$ with some $\bm\in\mathbb N_0^{3M}$. 
As an eigenfunction of $H$, the function $\psi$ is $\plainH2(\R^{3N})$, and, 
by elliptic regularity, 
it is 
smooth (even analytic) on the set  
$\mathcal S = \{x_k\not= 0, x_k \not= x_j:\ j, k = 1, 2, \dots, N\}$. 
Our objective is to show that 
the function $\psi$ has derivatives 
$u_\bm$ of all orders $|\bm|\ge 0$ 
on the larger set $U_\BP(0)\supset \mathcal S$, and that $u_\bm\in \plainH2(U_\BP(\varepsilon))$ 
for all $\varepsilon>0$.

Let us begin with a formal calculation. 
Since $H_E\psi = 0$, by Leibniz's formula, we obtain  
\begin{align}\label{eq:he}
H_E u_{\bm}  
= &\ [H_E, \SD_{\BP}^{\bm}]\psi = [V, \SD_{\BP}^{\bm}]\psi \notag\\
= &\ -\sum_{
\substack{0\le \bs \le \bm\\|\bs|\ge 1}}
 {\bm\choose\bs} \SD_{\BP}^{\bs} V\,  u_{\bm-\bs} = f_\bm. 
\end{align}
Thus $u_\bm$ is a solution of the equation $H_E u_\bm = f_\bm$. 
The next assertion gives this statement a precise meaning.  

First we observe that 
by Lemma \ref{lem:V}, $\|\SD_\BP^{\bs}V\|_{\plainL\infty(U_\BP(\varepsilon))}<\infty$ 
for every $\bs: |\bs|\ge 1,$ and all $\varepsilon>0$. 
Therefore, if $u_{\bm}\in\plainL2(U_\BP(\varepsilon))$ for 
all $\bm: |\bm|\le p$, then   
$f_\bm\in \plainL2(U_\BP(\varepsilon))$ for all $|\bm|\le p+1$. 

\begin{lem}\label{lem:weaksol}
Suppose that $u_\bm\in \plainL2(U_\BP(\varepsilon))$ for some $\varepsilon>0$ and all 
$\bm\in\mathbb N_0^{3M}$ 
such that $|\bm|\le p$ with some $p\in\mathbb N_0$. Then $u_\bm$ is a weak 
solution of the equation $H_E u_\bm = f_\bm$, that is, it 
satisfies the identity
\begin{align}\label{eq:weaksol}
\int u_{\bm} \overline{H_E \eta}d\bx = \int f_\bm\overline{\eta}d\bx,  
\end{align}
for all $\eta\in \plainC\infty_0(U_\BP(\varepsilon))$ and all $\bm: |\bm|\le p$.
\end{lem}

\begin{proof} 
As noted before the lemma, 
$f_\bm\in \plainL2(U_\BP(\varepsilon))$ for all $|\bm|\le p+1$, so that 
both sides of \eqref{eq:weaksol} are finite. Throughout the proof we use the fact that 
$\SD_\BP^{\bs}V\in \plainC\infty$ on $U_\BP(\varepsilon)$ for all $\bs: |\bs|\ge 1$, see 
Lemma \ref{lem:V}.

We prove the identity \eqref{eq:weaksol} by induction. 
First note that \eqref{eq:weaksol} 
holds for $\bm  = 0$, since $\psi$ is an eigenfunction and $f_0 = 0$. 
Suppose that it holds for all $\bm: |\bm|\le k,$ with some 
$k\le p-1$. We need to show that this implies \eqref{eq:weaksol} for 
$\bm+\bl$, where $\bl\in\mathbb N_0^{3M}: |\bl|=1$. 
As $u_{\bm+\bl} = \SD_\BP^{\bl}u_\bm$, we can integrate by parts, 
using \eqref{eq:weaksol} for $|\bm|\le k$:
\begin{align}\label{eq:plus1}
\int u_{\bm+\bl} \overline{H_E\eta} d\bx
= &\ - \int u_\bm \overline{\SD_\BP^{\bl} H_E\eta} d\bx 
= - \int u_\bm \overline{H_E\SD_\BP^{\bl}\eta} d\bx 
- \int u_\bm \big(\SD_\BP^{\bl}V\big) \overline\eta d\bx\notag\\
= &\ - \int f_\bm \overline{\SD_\BP^\bl\eta} d\bx 
- \int u_\bm \big(\SD_\BP^{\bl}V\big) \overline{\eta} d\bx.
\end{align}  
Integrating by parts and using definition of $f_{\bm}$ 
(see \eqref{eq:he}), we get for the first integral on the right-hand side that 
\begin{align*}
\int f_\bm \overline{\SD_\BP^\bl\eta}\, d\bx = \sum_{\bs:|\bs|\ge 1}^{\bm}{\bm\choose\bs} 
\int \big(
(\SD_\BP^{\bl+\bs} V) u_{\bm-\bs} + (\SD_\BP^{\bs}V) u_{\bm+\bl-\bs}
\big)\overline{\eta}\, d\bx. 
\end{align*}
Standard calculations involving binomial coefficients, show that 
%
%
\begin{align*}
\int f_\bm \overline{\SD_\BP^\bl\eta}\, d\bx = & \sum_{\bs:|\bs|\ge 1}^{\bm+\bl}{\bm+\bl\choose\bs} 
\int 
(\SD_\BP^{\bs} V) u_{\bm+\bl-\bs} \overline{\eta}\,d\bx 
- \int (\SD_\BP^{\bl}V) u_{\bm}\overline{\eta}\, d\bx\\
= &\ - \int f_{\bm+\bl} \overline{\eta} \, d\bx  - \int (\SD_\BP^{\bl}V) u_{\bm}\overline{\eta}\, d\bx.
\end{align*}
Substituting this into \eqref{eq:plus1}, we obtain that 
\begin{align*}
\int u_{\bm+\bl} \overline{H_E\eta} \,d\bx 
= 
%
%
\int f_{\bm+\bl}\overline{\eta} \,d\bx, 
\end{align*}
which coincides with \eqref{eq:weaksol} for $\bm+\bl$. Now by induction we conclude that \eqref{eq:weaksol} 
holds for all $\bm:|\bm|\le p$, as claimed.
\end{proof}

\begin{thm} 
Let $E$ be an eigenvalue of $H$ and let $\psi$ be the associated 
eigenfunction. 
For each $\varepsilon>0$ the function $u_\bm = \SD_\BP^{\bm}\psi$ 
belongs to $\plainH2(U_\BP(\varepsilon))$ for 
all $\bm\in \mathbb N_0^{3M}$. 
\end{thm}

\begin{proof} 
For brevity throughout the proof we use the notation 
$\CH^\a_\varepsilon = \plainH{\a}(U_\BP(\varepsilon))$, $\a = 1, 2$,  
$\CL^2_\varepsilon = \plainL{2}(U_\BP(\varepsilon))$.

The claim holds for $\bm =0$, since $\psi\in \plainH2(\R^{3N})$ is an eigenfunction 
and $f_0 = 0$. Suppose that it holds 
for all $\bm: |\bm|\le p\in \mathbb N_0$. We need to show that this implies that 
$u_{\bm+\bl}\in \CH^2_\varepsilon$, for all $\varepsilon>0$, where 
$\bl\in\mathbb N_0^{3M}:|\bl| = 1$ and $|\bm| = p$.

Since $u_\bm\in \CH^2_\varepsilon$, we have $u_{\bm+\bl}\in \CH^1_{\varepsilon}\subset \CL^2_\varepsilon$ 
for all $\varepsilon>0$. 
Thus, by Lemma \ref{lem:weaksol}, 
$u_{\bm+\bl}$ satisfies \eqref{eq:weaksol} with $f_{\bm+\bl}\in\CL^2_\varepsilon$. 
In order to show that $u_{\bm+\bl}\in\CH^2_\varepsilon$, for all $\varepsilon>0$, 
we apply Lemma \ref{lem:reg}. To this end 
let $\eta_1\in \plainC\infty(\R^{3N})$ 
be a function such that $\eta_1(\bx) = 0$ for $\bx\in \R^{3N}\setminus U_\BP(\varepsilon/2)$ and 
$\eta_1(\bx) = 1$ for $\bx\in U_\BP(\varepsilon)$. Thus, by \eqref{eq:weaksol}, 
\begin{align*}
H_E (u_{\bm+\bl} \eta_1) = &\ \eta_1 H_E u_{\bm+\bl} 
- 2\nabla\eta_1\nabla u_{\bm+\bl} - u_{\bm+\bl}\Delta \eta_1\\
= &\ \eta_1 f_{\bm+\bl} - 2\nabla\eta_1\nabla u_{\bm+\bl} - u_{\bm+\bl}\Delta \eta_1.
\end{align*}
Since $u_{\bm+\bl}\in \CH^1_{\varepsilon/2}$, the right-hand side belongs to $\plainL2(\R^{3N})$. 
Therefore, $H(u_{\bm+\bl}\eta)\in\plainL2(\R^{3N})$, and 
by Lemma \ref{lem:reg}, $u_{\bm+\bl}\eta_1\in \plainH2(\R^{3N})$. 
As a consequence, 
$u_{\bm+\bl}\in \CH^2_\varepsilon$, as required.     
Now, by induction, $u_\bm\in\CH^2_\varepsilon$ for all $\bm\in\mathbb N_0^{3M}$. 
\end{proof}

\subsection{Eigenfunction estimates} 
Apart from the qualitative fact of smoothness of $u_\bm = \SD_\BP^\bm\psi$, 
now we need to establish explicit estimates for $u_\bm$. 
As before we denote 
$H_E = H-E$ with an arbitrary $E\in\R$.

\begin{lem}\label{lem:coer}
Let $v\in \plainH2(U_\BP(\varepsilon))$ 
and let $\bm\in \mathbb N_0^{3N}$, $|\bm|\le 2$. Then for any $\varepsilon>0$, 
$\d\in (0, 1)$ we have 
\begin{align*}
\d^{|\bm|} \| \p^{\bm} v\|_{\plainL2(U_{\BP}(\varepsilon + \d))}
\le C_0 
\bigg(
\d^2 \|H v\|_{\plainL2(U_{\BP}(\varepsilon))} 
+ \max_{\substack{\bj\in\mathbb N_0^{3N}\\|\bj|\le 1}}\d^{|\bj|}
\|\p^{\bj} v\|_{\plainL2(U_{\BP}(\varepsilon))}
\bigg),
\end{align*}
with a constant $C_0$ 
independent of the function 
$v$, constants $\varepsilon$, $\d$ and of the cluster set $\BP$.  
\end{lem}

\begin{proof}
Let $|\bm|\le 1$. Since $U_\BP(\varepsilon+\d)\subset U_\BP(\varepsilon)$,  we have 
\begin{align*}
\d^{|\bm|} \| \p^{\bm} v\|_{\plainL2(U_{\BP}(\varepsilon + \d))}
\le \max_{\substack{\bj\in\mathbb N_0^{3N}\\|\bj|\le 1}}\d^{|\bj|}
\|\p^{\bj} v\|_{\plainL2(U_{\BP}(\varepsilon))}, 
\end{align*}
so that the required bound holds.

Assume now that $|\bm| = 2$. 
Without loss of generality assume that all clusters $\SP_s\in \BP$, 
$s = 1, 2, \dots, M$,  
are distinct.  
Let $\xi$ be the smooth function 
defined in \eqref{eq:xi}.  
For arbitrary $\varepsilon, \d>0$ define the cut-off 
\begin{align*}
\eta(\bx) = \eta_\BP(\bx)
= \prod_{s=1}^M 
\prod_{\substack{k\in \SP_s\\j\in \SP_s^{\rm c}}}
\xi\bigg(\frac{|x_k|-\varepsilon}{\d}\bigg)
\xi\bigg(\frac{|x_k-x_j|-\varepsilon}{\d}\bigg).
\end{align*}
Then $\supp\eta\subset U_\BP(\varepsilon)$ and $\eta = 1$ on $U_\BP(\varepsilon+\d)$. 
It is also clear that
\begin{align*}
\max_{\BP}\|\p^{\bk}\eta\|\le C_{\bk}|\d|^{-|\bk|},\quad \forall \bk\in\mathbb N_0^{3N},
\end{align*}
with some positive constants $C_\bk$ independent of $\varepsilon$ and $\d$, 
where the maximum is taken over all sets $\BP$ of distinct clusters. 
Estimate, using the bound 
\eqref{eq:domain}:
\begin{align*}
\| \p^{\bm} v\|_{\plainL2(U_\BP(\varepsilon + \d))}
\le &\ \| \p^{\bm} (v\eta)\|_{\plainL2}
\le C\big(\|H (v\eta)\|_{\plainL2} + \|v\eta\|_{\plainL2}\big)\\ 
\le &\ C\big(\|\eta H v\|_{\plainL2} 
+ \|v\Delta\eta\|_{\plainL2} + 2\|(\nabla\eta)\nabla v\|_{\plainL2}
+ \|v\eta\|_{\plainL2} 
\big)\\
\le &\ \tilde C\big(\|H  v\|_{\plainL2(U_\BP(\varepsilon))} 
+ (\d^{-2}+1)\|v\|_{\plainL2(U_\BP(\varepsilon))}  
+ \d^{-1}\|\nabla v\|_{\plainL2(U_\BP(\varepsilon))} \big),
\end{align*}
with constants independent of $\varepsilon, \delta$.
Multiplying by $\d^2$, we get the required estimate.
\end{proof}

%
%

Let $E$ be an eigenvalue of $H$ and $\psi$ be the associated eigenfunction. 
Now we use Lemma \ref{lem:coer} for the function 
$v = u_\bm = \SD_\BP^\bm\psi\in \plainH2\big(U_\BP(\varepsilon)\big), \varepsilon>0$. 

\begin{cor}
There exists a constant $L_2>0$ 
independent of the cluster set $\BP$ and of the parameters $\varepsilon>0,\d\in (0, 1)$, 
such that for all 
$\bm\in \mathbb N_0^{3M}$, $\bk, \bl\in \mathbb N_0^{3N}$, $|\bk|+ |\bl|\le 2$,  
 we have 
\begin{align}\label{eq:coer}
\d^{|\bk| + |\bl|} \| \p^{\bk}  
\SD_\BP^{\bm+\bl}\psi
\|_{\plainL2(U_\BP(\varepsilon + \d))}
\le L_2\bigg(
\d^2 \|f_\bm\|_{\plainL2(U_\BP(\varepsilon))} 
+ \max\limits_{\substack{\bj\in\mathbb N_0^{3N}\\|\bj|\le 1}}\d^{|\bj|}
\|\p^{\bj} \SD_\BP^{\bm}\psi\|_{\plainL2(U_\BP(\varepsilon))}
\bigg).
\end{align}
\end{cor}

\begin{proof}
Apply Lemma \ref{lem:coer} to the function $v= u_\bm$ and estimate 
\begin{align*}
\|H u_\bm\|_{\plainL2(U_\BP(\varepsilon))} 
\le &\ \| H_Eu_\bm\|_{\plainL2(U_\BP(\varepsilon))} + |E|\|u_\bm\|_{\plainL2(U_\BP(\varepsilon))} \\
= &\ \|f_\bm\|_{\plainL2(U_\BP(\varepsilon))} + |E|\|u_\bm\|_{\plainL2(U_\BP(\varepsilon))}.
\end{align*}
\end{proof}

Now we use the bound \eqref{eq:coer} to obtain estimates 
for the function $u_{\bm}$ 
with arbitrary $\bm\in\mathbb N_0^{3M}$. 
Let $A_0$, $L_2$ and $L_3$ be the constants featuring in 
\eqref{eq:V}, \eqref{eq:coer} and \eqref{eq:choosepow1} respectively. Define 
\begin{align}\label{eq:L0}
A_1 = 2 A_0 + L_2(L_3 A_0 + 1) + \max\limits_{\bj: |\bj|\le 1}
\|\p^{\bj} \psi\|_{\plainL2(\R^{3N})}.
\end{align}
Thus defined constant depends on the eigenvalue $E$ and $\varepsilon>0$, but is independent 
of the cluster set $\BP$ and of $\delta\in (0, 1]$. 

\begin{lem} 
Let the constant $A_1$ be as defined in \eqref{eq:L0}. Then 
for all $\bm\in \mathbb N_0^{3M}$, all 
$\bk\in\mathbb N_0^{3N}$, $|\bk|\le 1$, 
and all $\varepsilon>0$ and  $\d>0$ such that $\d(|\bm|+1)\le 1$, we have 
\begin{align}\label{eq:induction}
\|\p^{\bk}\SD_{\BP}^{\bm}\psi\|_{\plainL2(U_{\BP}(\varepsilon+(|\bm|+1)\d)} 
\le A_1^{|\bm|+1} \d^{-|\bm| - |\bk|}.
\end{align}
\end{lem}

\begin{proof}
The formula \eqref{eq:induction} holds for 
$\bm = \mathbf 0$. Indeed, since $\d\le 1$, we get 
\begin{align*}
\d^{|\bk|}\|\p^{\bk} \psi\|_{\plainL2(U_{\BP}(\varepsilon+\d))}
\le  
\max\limits_{\bj: |\bj|\le 1}\d^{|\bj|}
\|\p^{\bj} \psi\|_{\plainL2(\R^{3N})} \le A_1.
\end{align*}
Further proof is by induction. 
As before, we use the notation $u_\bm = \SD_\BP^{\bm}\psi$. 
Suppose that \eqref{eq:induction} holds for all 
$\bm\in\mathbb N_0^{3M}$ such that $|\bm|\le p$ with some $p$. 
Our task is to deduce from this that \eqref{eq:induction} holds for all $\bm$, such that $|\bm| = p+1$. 
Precisely, we need to show that if $|\bm|=p$ and $\bl\in \mathbb N_0^{3M}$ is 
such that $|\bl|=1$, then  
\begin{align}\label{eq:induction1}
\|\p^{\bk} u_{\bm+\bl}
\|_{\plainL2(U_{\BP}(\varepsilon+(p+2)\d)} 
\le A_1^{p+2} \d^{-p - |\bk|-1},
\end{align}
for all $\d>0$ such that $(p+2)\d \le 1$.

Since $|\bl|+|\bk| = 1+|\bk|\le 2$, it follows from \eqref{eq:coer} that    
\begin{align}\label{eq:ind}
\d^{|\bk|+1}\|\p^{\bk} u_{\bm + \bl} 
\|_{\plainL2(U_{\BP}(\varepsilon+ (p+2)\d))}
\le &\ L_2 \bigg(
\d^{2} \| f_\bm\|_{\plainL2(U_{\BP}(\varepsilon+(p+1)\d)}\notag\\
&\quad\quad\quad\quad + \max_{\substack{\bj\in\mathbb N_0^{3N}\\|\bj|\le 1}}\d^{|\bj|}
\|\p^{\bj}u_\bm\|_{\plainL2(U_{\BP}(\varepsilon+(p+1)\d))}
\bigg). 
\end{align}
By the induction hypothesis, the second term in the brackets on the right-hand side 
satisfies the bound 
\begin{align}\label{eq:second}  
\max_{\substack{\bj\in\mathbb N_0^{3N}\\|\bj|\le 1}}\d^{|\bj|}
\|\p^{\bj}u_{\bm}\|_{\plainL2(U_{\BP}(\varepsilon+(p+1)\d))}
\le \d^{-p}A_1^{p+1}.  
\end{align}
Let us estimate the first term on the right-hand side of \eqref{eq:ind}.  
First we find suitable bounds for the norms of the functions $u_{\bm-\bs}$, $0\le \bs\le \bm$, $|\bs|\ge 1$, 
featuring in the definition of the function $f_\bm$, see \eqref{eq:he}. Denote $q = |\bs|$. 
Since $|\bm-\bs|\le p$, we can use the induction assumption to obtain 
\begin{align*}
\|u_{\bm-\bs}\|_{\plainL2(U_\BP(\varepsilon+(p-q+1)\tilde\d))}\le A_1^{p-q+1}{\tilde\d}^{-p+q},
\end{align*}
for all $\tilde\d$ such that $(p-q+1)\tilde\d\le 1$. In particular, the value 
$\tilde\d = (p+1)(p-q+1)^{-1}\d$ satisfies the latter requirement, because $(p+1)\d\le 1$. Thus 
\begin{align*}
\|u_{\bm-\bs}\|_{\plainL2(U_{\BP}(\varepsilon+(p+1)\d))}
\le A_1^{p-q+1} (p+1)^{-p+q}(p-q+1)^{p-q} \d^{-p+q}.
\end{align*}
For the derivatives of $V$ we use \eqref{eq:V}, so that
\begin{align*}
\|\SD_{\BP}^{\bs} V\|_{\plainL\infty(U_{\BP}(\varepsilon+(p+1)\d))}
\le 
\|\SD_{\BP}^{\bs} V\|_{\plainL\infty(U_{\BP}(\varepsilon))}
\le  A_0^{q+1} (q+1)^q.
\end{align*}
Using the definition of $f_\bm$, see 
\eqref{eq:he}, and 
putting together the two previous estimates, we obtain
\begin{align*}
\|f_\bm\|_{\plainL2(U_{\BP}(\varepsilon+(p+1)\d))}
\le \sum_{q=1}^p \sum_{|\bs| = q}{\bm\choose\bs} A_0^{q+1} (q+1)^q  
A_1^{p-q+1} (p+1)^{-p + q}(p-q+1)^{p-q}
\d^{-p+q}. 
\end{align*}
In view of \eqref{eq:choose}, the right-hand side coincides with 
\begin{align*}
A_0 A_1^{p+1}\sum_{q = 1}^p {p\choose q} \big(A_0 A_1^{-1}\big)^q (q+1)^q 
(p+1)^{-p + q}(p-q+1)^{p-q}
\d^{-p+q}.
\end{align*}
Estimate the coefficient ${p\choose q}$, using \eqref{eq:choosepow1}: 
\begin{align*}
\|f_\bm & \|_{\plainL2(U_{\BP} (\varepsilon+(p+1)\d))}\\
\le &\ 
L_3 A_0 A_1^{p+1}\d^{-p}\sum_{q=1}^p   \big(A_0 A_1^{-1}\big)^q ((1+p)\d)^q
\le L_3 A_0 A_1^{p+1}\d^{-p}\sum_{q=1}^p   \big(A_0 A_1^{-1}\big)^q, 
\end{align*}
where we have taken into account that $(p+1)\d\le 1$. 
 By \eqref{eq:L0}, we have $A_0 A_1^{-1}\le 1/2$, 
 so that the sum on the right-hand side does not exceed $1$. 
Since $\d \le 1$, we can now conclude that 
\begin{align*}
\d^2\| f_\bm\|_{\plainL2(U_{\BP}(\varepsilon+(p+1)\d))}
\le  L_3 A_0 A_1^{p+1} \d^{-p+2}\le L_3 A_0 A_1^{p+1} \d^{-p}.
\end{align*}
Substituting 
this bound together with \eqref{eq:second} in 
\eqref{eq:ind} we arrive at the estimate  
\begin{align*}
\|\p^{\bk} u_{\bm+\bl}\|_{\plainL2(U_{\BP}(\varepsilon+ (p+2)\d)}
\le &\ \d^{-p-1 -|\bk|} A_1^{p+1} L_2\big(1 + L_3 A_0\big). 
\end{align*}
By the definition \eqref{eq:L0}, the factor $L_2(1 + L_3 A_0)$ does not exceed $A_1$. 
This leads to the bound \eqref{eq:induction1}, and hence proves the lemma. 
\end{proof}

\begin{cor}\label{cor:derivl2} 
For any $\varepsilon\in (0, 1]$ there is a constant 
$A_2=A_2(\varepsilon)$, such that 
for all cluster sets 
$\BP = \{\SP_1, \SP_2, \dots, \SP_M\}$ and all $\bm\in\mathbb N_0^{3M}$, we have 
\begin{align*}
\|\SD^{\bm}_{\BP}\psi\|_{\plainL2({U}_{\BP}(2\varepsilon))}
\le A_2^{|\bm|+1} (1+|\bm|)^{|\bm|}.
\end{align*}
\end{cor}

\begin{proof}
Use \eqref{eq:induction} with $\bk = \mathbf 0$ and $\d = (|\bm|+1)^{-1}\varepsilon$:
\begin{align*}
\|\SD_{\BP}^{\bm} \psi\|_{\plainL2({U}_{\BP}(2\varepsilon))}
\le \varepsilon^{-|\bm|}A_1^{|\bm|+1} (1+|\bm|)^{|\bm|}
\le A_2^{|\bm|+1} (1+|\bm|)^{|\bm|}, 
\end{align*}
with $A_2 = \varepsilon^{-1} A_1$, where we have taken into account that $\varepsilon\le 1$.
\end{proof}

\section{Cut-off functions and associated clusters}\label{sect:cutoff}

\subsection{
Admissible 
cut-off functions}
Let $\{f_{jk}\}, 1\le j, k\le N$, be a set of functions such that 
each of them is one of the functions $\z, \t$ or $\p_x^l \t$, $l\in \mathbb N_0^3, |l|=1$, and 
$f_{jk} = f_{kj}$.
We work with the smooth functions of the form
\begin{align}\label{eq:canon}
\phi(\bx) = \prod\limits_{1\le j < k\le N} f_{j k}(x_j-x_k).
\end{align}
We call such functions \textit{admissible cut-off functions}
or simply \textit{admissible cut-offs}. 
For any such function 
$\phi$ we also introduce the following ``partial" products. 
For an arbitrary cluster 
$\SP\subset \SR = \{1, 2, \dots, N\}$ 
define  
\begin{align*}
\phi(\bx; \SP) = 
\begin{cases}
\prod\limits_{\substack{j < k\\
j, k\in \SP}}f_{j k}(x_j-x_k),\quad &{\rm if}\ |\SP|\ge 2;\\[0.2cm]
\qquad 1,\quad &{\rm if}\ |\SP|\le 1.
\end{cases}
\end{align*}
Furthermore, for any two 
clusters $\SP, \SfS\subset \SR$, such that $\SfS\cap \SP = \varnothing$, we define
\begin{align}\label{eq:PS}
\phi(\bx; \SP, \SfS) = 
\begin{cases}
\prod\limits_{j\in \SP, k\in \SfS}f_{j k}(x_j-x_k),
\quad &{\rm if}\ \SP\not = \varnothing \ {\rm and}\ \SfS\not=\varnothing;\\[0.2cm]
\qquad 1,\quad &{\rm if}\ \SP = \varnothing\ {\rm or}\ \SfS = \varnothing.
\end{cases}
\end{align} 
Note that 
\begin{align}\label{eq:diffl}
\SD_{\SP}^l \phi(\bx; \SP) = \SD_{\SP}^l \phi(\bx; \SP^{\rm c}) = 0,\quad \textup{for all}
\ l\in\mathbb N_0^3,\, |l|\ge 1.
\end{align}
It is straightforward to see that for any cluster $\SP$ the function 
$\phi(\bx)$ can be represented as follows:
\begin{align}\label{eq:phirep}
\phi(\bx) = \phi(\bx; \SP) \phi(\bx; \SP^{\rm c}) 
\phi(\bx; \SP, \SP^{\rm c}).
\end{align} 
Following 
\cite{FHOS2004}, we associate with the function $\phi$ a cluster $\SQ(\phi)$ defined next. 

\begin{defn} 
For an admissible cut-off $\phi$, let $I(\phi)\subset \{(j, k)\in\SR\times\SR: j\not= k\}$
be the set such that 
$(j, k)\in I(\phi)$, iff $f_{jk}\not=\t$.    
We say that two indices $j, k\in\SR$, 
are $\phi$-\textit{linked} to each other 
if either $j=k$, or $(j, k)\in I(\phi)$, or 
there exists a sequence of pairwise 
distinct indices $j_1, j_2, \dots, j_s$, $1\le s\le N-2$, 
all distinct from $j$ and $k$, 
such that $(j, j_1), (j_s, k)\in I(\phi)$ and $(j_p, j_{p+1})\in I(\phi)$ for all $p = 1, 2, \dots, s-1$.

The cluster $\SQ(\phi)$ is defined as the set of all indices that are 
$\phi$-linked to index $1$.  
\end{defn}

It follows from the above definition that $\SQ(\phi)$ always contains index $1$. Note also that 
the notion of being linked defines an equivalence relation on $\SR$, and the cluster $\SQ(\phi)$ is 
nothing but the equivalence class of index $1$.

For example, if $N = 4$ and 
\begin{align*}
\phi(\bx) = \z(x_1-x_2) \t(x_1 - x_3)
\t(x_1-x_4)\p_x^l\t(x_2-x_3)  \t(x_2 - x_4)\t(x_3-x_4),
\end{align*}
with some $l\in \mathbb N_0^3$, $|l|=1$, 
then $\SQ(\phi) = \{1, 2, 3\}$.
 
Let $\SP = \SQ(\phi)$. 
If $\SP^{\rm c}$ is not empty, i.e. 
$\SP\not = \SR$, then, by the definition of $\SP$, 
we always have $f_{jk}(x) = \t(x)$ for all $j\in \SP$ and $k\in \SP^{\rm c}$, and hence 
the representation \eqref{eq:phirep} holds with 
\begin{align}\label{eq:phipq}
\phi(\bx; \SP, \SP^{\rm c}) = \prod_{j\in \SP,\  k\in \SP^{\rm c}} \t(x_j-x_k).
\end{align} 
The notion of associated cluster is useful because of its connection with the 
support of the cut-off $\phi$. This is clear from the next two lemmata. 
Recall that the sets $X_{\SP}, \widehat T_{\SP}$ are defined in 
 \eqref{eq:xp} and \eqref{eq:hattp} respectively.

\begin{lem}\label{lem:supp} 
For $\SP = \SQ(\phi)$ the inclusion 
\begin{align}\label{eq:suppphi}
\supp\phi
\subset X_{\SP}\big(\varepsilon(4N)^{-1}\big)
\end{align}
holds. 
\end{lem}

\begin{proof} 
If $\SP^{\rm c} = \varnothing$, then, by definition,  $X_{\SP} = \R^{3N}$, and hence 
\eqref{eq:suppphi} is trivial. 

Suppose that $\SP^{\rm c}$ is non-empty. 
The inclusion \eqref{eq:suppphi} immediately 
follows from the representation \eqref{eq:phirep}, formula \eqref{eq:phipq} 
and the definition of the function $\t$. 
\end{proof}

\begin{lem}\label{lem:link}
If $j\in \SQ(\phi)$, then $|x_1-x_j|<\varepsilon/2$ for all $\bx\in\supp\phi$. 

Moreover, 
\begin{align}\label{eq:suppphi1}
\supp\phi(x_1, \ \cdot\ )
\subset \widehat T_{\SP^*}\big(\varepsilon/2\big),
\end{align}
for all $x_1: |x_1|>\varepsilon$.
\end{lem}

\begin{proof}
Let $\bx\in \supp\phi$. 
By the definition of $\z$ and $\t$, if $(j, k)\in I(\phi)$, then $|x_j-x_k|< \varepsilon(2N)^{-1}$. 
Thus, if $j$ and $k$ are $\phi$-linked to each other, then  
\begin{align*}
|x_j-x_k|\le &\ |x_j - x_{j_1}| 
+ \sum_{p=1}^{s-1} |x_{j_p} - x_{j_{p+1}}| + |x_{j_s} - x_k| \notag\\
\le &\ \frac{\varepsilon}{2N}(s+1) < \frac{\varepsilon}{2}.
\end{align*} 
In particular, for $j\in \SQ(\phi)$ we have $|x_1 - x_j|< \varepsilon/2$, 
as claimed.

Proof of \eqref{eq:suppphi1}. Suppose that $\bx\in\supp\phi$ and $|x_1| > \varepsilon$. 
By Lemma \ref{lem:link}, for each $j\in \SP^*$ we have $|x_1-x_j|<\varepsilon/2$, so that 
\begin{align*}
|x_j|\ge |x_1| - |x_1 - x_j|> \frac{\varepsilon}{2},
\end{align*}
 as claimed.
\end{proof}

To summarize in words, on the support of the admissible cut-off $\phi$ the variables 
$x_j$, indexed by $j\in\SP = \SQ(\phi)$, are ``close" to each other and  
``far" from the remaining variables. 
 
Let $\phi$ be of the form \eqref{eq:canon}, and let $\SP = \SQ(\phi)$. 
For each $l\in \mathbb N_0^3$ the function $\SD_{\SP}^l\phi$ has the form 
\begin{align}\label{eq:phider}
 \SD_{\SP}^l \phi(\bx) = 
\phi(\bx; \SP) \phi(\bx; \SP^{\rm c}) \SD_{\SP}^l \phi(\bx; \SP, \SP^{\rm c}),\quad \SP = \SQ(\phi),
\end{align}
where we have used the factorization \eqref{eq:phirep} and property \eqref{eq:diffl}. 
By the definition \eqref{eq:PS}, $\SD_{\SP}^l\phi = 0$ if $\SP^{\rm c} = \varnothing$. 
 
\begin{lem}
Let $\SP = \SQ(\phi)$ and $l\in\mathbb N_0^3, |l|=1$. 
If $\SP^{\rm c}\not = \varnothing$, then the function \eqref{eq:phider} is represented in the form
\begin{align}\label{eq:derivrep}
\SD_{\SP}^l \phi(\bx)  
= \sum\limits_{s\in \SP, r\in \SP^{\rm c}} \phi_{s, r}^{(l)},
\end{align}
where each $\phi^{(l)}_{s, r}$ is an admissible cut-off of the form 
\begin{align}\label{eq:philrs}
\phi
^{(l)}_{s, r}(\bx)
= 
\phi(\bx; \SP) \phi(\bx; \SP^{\rm c}) \p_x^l \t(x_s-x_r)
\prod_{\substack{j\in \SP, k\in \SP^{\rm c}\\(j,k)\not = (s, r)}} \t(x_j-x_k). 
\end{align}
Moreover, $\SP\subset \SQ(\phi^{(l)}_{s, r})$ and 
$|\SQ(\phi_{s, r}^{(l)})| \ge  |\SP|+1$.
\end{lem} 

\begin{proof} 
The representation \eqref{eq:derivrep} immediately follows 
from the definition \eqref{eq:phipq}. 
It is clear from 
\eqref{eq:philrs} that $\phi^{(l)}_{s, r}$ has the form 
\eqref{eq:canon}, and hence it is admissible.

Due to the presence of the derivative $\p_x^l\t$, in addition to all 
indices linked to index $1$ by the function $\phi$, the new function $\phi^{(l)}_{r, s}$ links 
the indices $r$ and $s$ as well, and hence its associated 
cluster $\SQ\big(\phi^{(l)}_{r, s}\big)$ 
contains $\SP$ and $|\SQ\big(\phi^{(l)}_{r, s}\big)|\ge  |\SP|+1$, as claimed. 
\end{proof}
 
In what follows a special role is played by the factorization \eqref{eq:phirep} with 
$\SP = \{1\}$, so that 
\begin{align}\label{eq:factor}
\phi(x_1, \hat\bx) = \om(x_1, \hat\bx)\varkappa(\hat\bx)
\quad {\rm with}\quad \om(x_1, \hat\bx) = \phi(x_1, \hat\bx; \{1\}, \SR^*),\ 
\varkappa(\hat\bx) = \phi(\bx; \SR^*).
\end{align}
We call the functions $\om$ and $\varkappa$ the \textit{canonical} factors of $\phi$.  
In the next corollary we find the canonical factors for the cut-offs $\phi^{(l)}_{s, r}$ defined 
in \eqref{eq:philrs}.

\begin{cor}
Let $\om, \varkappa$ be the canonical factors of $\phi$, and let $\SP^{\rm c}\not = \varnothing$. 
Then the functions $\phi^{(l)}_{s, r}$ can be represented as follows:
\begin{align*}
\phi^{(l)}_{s, r}(x_1, \hat\bx) = \om^{(l)}_{r, s} (x_1, \hat\bx)
\varkappa^{(l)}_{s, r}(\hat\bx),\quad  
s\in \SP, r\in \SP^{\rm c},
\end{align*}
with 
\begin{align}\label{eq:factor1r}
\begin{split}
\om^{(l)}_{1, r}(x_1, \hat\bx) = &\ \phi(x_1, \hat\bx; \{1\}, \SP^*)
\p_x^l \t(x_1-x_r)\prod_{k\in\SP^{\rm c}, k\not = r}
\t(x_1-x_k),\\
\varkappa^{(l)}_{1, r}(\hat\bx) = &\ \varkappa(\hat\bx),
\end{split}
\end{align}
and 
\begin{align}\label{eq:factorsr}
\begin{split}
\om^{(l)}_{s, r}(x_1, \hat\bx) = &\ \om(x_1, \hat\bx),\\
\varkappa^{(l)}_{s, r}(\hat\bx) =  &\ \phi(\hat\bx; \SP^*) \phi(\hat\bx; \SP^{\rm c}) 
\p_x^l \t(x_s-x_r)
\prod_{\substack{j\in \SP^*, k\in \SP^{\rm c}\\(j,k)\not = (s, r)}} \t(x_j-x_k),
\end{split}
\end{align}
for all $s\in \SP^*$. 
\end{cor}

\begin{proof} 
The claim is an immediate consequence of \eqref{eq:philrs}. 
\end{proof}

\subsection{Extended cut-offs} \label{subsect:extended} 
Now we are ready to introduce the cut-off functions with which we work when estimating  
the derivatives of the one-particle density matrix $\g(x, y)$. These cut-offs are  
functions of $3N+3$ variables and they are defined as follows. 
We say that two admissible cut-offs 
$\phi = \phi(x_1, \hat\bx)$ and $\mu = \mu(x_1, \hat\bx)$ 
are \textit{coupled 
to each other} if they share the same canonical factor  
$\varkappa = \varkappa(\hat\bx)
= \phi(\bx; \SR^*) = \mu(\bx; \SR^*)$, i.e.
\begin{align*}
\phi(x_1, \hat\bx) = \om(x_1, \hat\bx) \varkappa(\hat\bx),\quad 
\mu(x_1, \hat\bx) = \tau(x_1, \hat\bx) \varkappa(\hat\bx), 
\end{align*}
where $\om$ is defined as in \eqref{eq:factor} and 
$\tau(x_1, \hat\bx) = \mu(x_1, \hat\bx; \{1\}, \SR^*)$.
Out of two coupled cut-offs $\phi, \mu$ we construct a new 
function of $3N+3$ variables:  
\begin{align}\label{eq:Phi}
\Phi(x, y, \hat\bx) = &\ \om(x, \hat\bx) \tau(y, \hat\bx) \varkappa(\hat\bx)
\notag\\ 
= &\ \phi(x, \hat\bx)\tau(y, \hat\bx) 
= \om(x, \hat\bx) \mu(y, \hat\bx),\quad (x, y)\in\R^3\times\R^3,   
\hat\bx\in\R^{3N-3}.
\end{align}
 We call such $\Phi$ an \textit{extended} cut-off. 
It is clear that every extended cut-off defines a pair of coupled admissible $\phi$ and $\mu$ 
uniquely. We say that the pair $\phi$, $\mu$ and the extended 
cut-off $\Phi$ are associated to each other. 
The representations \eqref{eq:Phi} and identity \eqref{eq:supp} 
give the equality 
\begin{align}\label{eq:suppPhib}
\supp\Phi(x, y,\ \cdot\ )
=  
\supp\phi(x, \ \cdot\ )\cap \supp\mu(y, \ \cdot\ ),\quad \forall (x, y)\in\R^3\times\R^3.
\end{align} 
From now on we denote $\SP = \SQ(\phi)$ and $\SfS = \SQ(\mu)$.
 
Below we list some useful 
properties of the extended cut-offs $\Phi$ and associated admissible $\phi$, $\mu$. Due 
to the nature of the definition \eqref{eq:Phi}, in 
all statements involving the functions $\phi$, 
$\mu$ and $\Phi$, the pairs $\{\phi, \SP\}$ and $\{\mu, \SfS\}$ can be interchanged.

\begin{lem}\label{lem:empty}
If $\SP^*\cap \SfS $ is non-empty, 
then $\Phi(x, y, \hat\bx) = 0$ for all $\hat\bx\in\R^{3N-3}$ if $(x, y)\in \CD_\varepsilon$. 
\end{lem}

\begin{proof}
Suppose that $\SP^*\cap\SfS$ is non-empty and that 
$(x, \hat\bx)\in\supp\phi$, $(y, \hat\bx)\in\supp\mu$. 
By Lemma \ref{lem:link}, for each $j\in \SP^*\cap \SfS$ we have 
$|x-x_j|<\varepsilon/2$ and $|y-x_j|<\varepsilon/2$.  
Hence $|x-y|<\varepsilon$, and so 
\begin{align*}
\supp\phi(x, \ \cdot\ )\cap \supp\mu(y, \ \cdot\ ) = \varnothing,\ \quad{\rm if}\ (x, y)\in \CD_\varepsilon.
\end{align*}
By \eqref{eq:suppPhib}, $\Phi(x, y, \hat\bx) = 0$ for 
all $\hat\bx\in\R^{3N-3}$ if $(x, y)\in\CD_\varepsilon$, as claimed. 
\end{proof}

Due to Lemma \ref{lem:empty} from now on  we may assume that $\SP^*\subset\SfS^{\rm c}$.

\begin{lem}
Let $\phi$ and $\mu$ be coupled admissible cut-offs, and let $\SP^*\subset \SfS^{\rm c}$. Then 
\begin{align}\label{eq:murep}
\tau(x_1, \hat\bx) = 
\mu\big(x_1, \hat\bx; \{1\}, \SP^c\big) \prod\limits_{j\in \SP^*} \t(x_1-x_j),
\end{align}
and
\begin{align}\label{eq:suppmu}
\supp\mu\subset X_{\SP^*}(\varepsilon(4N)^{-1}).
\end{align}
Furthermore, 
\begin{align}\label{eq:cut}
\tau(y, \hat\bx) = \mu\big(y, \hat\bx; \{1\}, \SP^c\big), 
\quad \textup{if}\ (x, y)\in \CD_\varepsilon\quad \textup{and} \ (x, \hat\bx)\in \supp\phi.
\end{align}
\end{lem} 

\begin{proof} 
  Since $\SR^* = \SP^*\cup \SP^{\rm c}$, the function 
$\tau(x_1, \hat\bx) = \mu(x_1, \hat\bx;\{1\}, \SR^*)$  
factorizes as follows:
\begin{align*}
\tau(x_1, \hat\bx) = \mu\big(x_1, \hat\bx; \{1\}, \SP^c\big)
\mu\big(x_1, \hat\bx; \{1\}, \SP^*\big).
\end{align*}
As $\SP^*\subset \SfS^{\rm c}$, we have 
\begin{align*}
\mu(x_1, \hat\bx; \{1\}, \SP^*) = \prod\limits_{j\in \SP^*} \t(x_1-x_j),
\end{align*}
which leads to \eqref{eq:murep}. 

If $\SP^{\rm c} = \varnothing$, i.e. $\SP = \SR$, 
then the inclusion $\SP^*\subset \SfS^{\rm c}$ implies that 
$\SfS^{\rm c} = \SP^*$. 
By Lemma \ref{lem:supp}, 
\begin{align*}
\supp \mu\subset X_{\SfS}(\varepsilon(4N)^{-1}).
\end{align*} 
As $X_\SfS = X_{\SfS^{\rm c}}$, the claimed result follows. 

Assume now that $\SP^{\rm c}\not = \varnothing$. 
Consider separately the factors in the representation 
$\mu(\bx) = \tau(x_1, \hat\bx)\varkappa(\hat\bx)$.
Since 
\begin{align*}
\varkappa(\hat\bx) = \phi(\hat\bx; \SP^*) 
\phi(\hat\bx; \SP^{\rm c}) \phi(\hat\bx; \SP^*, \SP^{\rm c}), 
\end{align*}
in view of 
\eqref{eq:phipq} and definition \eqref{eq:pu} we have 
\begin{align}\label{eq:suppkappa}
\supp\varkappa\subset &\ \supp \phi(\ \cdot\ ; \SP^*, \SP^{\rm c})\notag\\
\subset&\ \{\hat\bx: |x_j-x_k|> \varepsilon(4N)^{-1}, 
j\in \SP^*, k\in \SP^{\rm c}\}.
\end{align} 
In view of \eqref{eq:murep}, by the definition \eqref{eq:pu} again,
\begin{align*}
\supp\tau\subset\{\bx: |x_1-x_j|>\varepsilon(4N)^{-1}, j\in\SP^*\}.
\end{align*} 
Since $\SP^{\rm c}\cup\{1\} = (\SP^*)^{\rm c}$, 
together with \eqref{eq:suppkappa}, this gives the inclusion
\begin{align*}
\supp\mu = \supp\varkappa\tau\subset \{\bx: |x_j-x_k|> \varepsilon(4N)^{-1}, 
j\in \SP^*, k\in (\SP^*)^{\rm c}\} = X_{\SP^*}(\varepsilon(4N)^{-1}),
\end{align*}
as required. 

Proof of \eqref{eq:cut}. Since $(x, \hat\bx)\in\supp\phi$, 
by Lemma \ref{lem:link} we have $|x-x_j|<\varepsilon/2$, $j\in\SP^*$. 
As $(x, y)\in\CD_\varepsilon$, this implies that  
$|y-x_j|> \varepsilon/2$ for all $j\in\SP^*$. 
By definition \eqref{eq:pu}, $\t(y-x_j) = 1$ if $|y-x_j|>\varepsilon(2N)^{-1}$. 
Therefore  $\t(y-x_j)=1$ for all 
$j\in\SP^*$. Hence the product on the right-hand side of 
\eqref{eq:murep} equals one, which implies \eqref{eq:cut}.
\end{proof}

In the next lemma we collect all the information on the supports 
of the coupled $\phi, \mu$ and associated extended 
cut-off $\Phi$ that we need in the next section. 

\begin{lem} 
If $\SP^*\subset \SfS^{\rm c}$, then 
\begin{align}\label{eq:phimux}
\begin{cases}
\supp \phi\subset X_{\SP}\big(\varepsilon(4N)^{-1}\big)\cap X_{\SfS^*}\big(\varepsilon(4N)^{-1}\big),\\[0.2cm]
\supp \mu\subset X_{\SfS}\big(\varepsilon(4N)^{-1}\big)\cap X_{\SP^*}\big(\varepsilon(4N)^{-1}\big).
\end{cases}
\end{align}
If $(x, y)\in \CD_\varepsilon$, then 
\begin{align}\label{eq:phihat}
\supp\Phi(x, y, \ \cdot\ ) = \supp\phi(x, \ \cdot\ )\cap \supp\mu(y, \ \cdot\ )\subset \widehat 
T_{\SP^*}(\varepsilon/2)\cap \widehat T_{\SfS^*}(\varepsilon/2).
\end{align}
\end{lem}  
 
\begin{proof} 
The second inclusion in \eqref{eq:phimux} follows from the inclusion \eqref{eq:suppphi} 
applied to function $\mu$ and 
from the relation \eqref{eq:suppmu}. The condition $\SP^*\subset\SfS^{\rm c}$ 
is equivalent to $\SfS^*\subset\SP^{\rm c}$. Thus the first inclusion in 
\eqref{eq:phimux} follows from the inclusion \eqref{eq:suppphi} 
and the relation \eqref{eq:suppmu} applied to 
the function $\phi$.

Since $|x|>\varepsilon$ and $|y|>\varepsilon$ for $(x, y)\in \CD_\varepsilon$, 
the inclusion \eqref{eq:phihat} follows from \eqref{eq:suppphi1} applied to $\phi$ and $\mu$. 
\end{proof} 

Similarly to the cluster derivatives 
\eqref{eq:phider} of the admissible cut-offs, 
now we need to investigate the cluster derivatives of the extended cut-offs. 
As before, let $\Phi(x, y, \hat\bx)$ be an extended cut-off associated with the coupled 
admissible cut-offs $\phi, \mu$, and let $\SP = \SQ(\phi)$, $\SfS = \SQ(\mu)$. 
It will be sufficient to confine our attention to the derivatives $\SD_\SP^l$ 
w.r.t. the variable $\bx$ 
for the cluster $\SP$. 

Assume that $\SP^{\rm c}\not = \varnothing$. 
Let $\phi^{(l)}_{s, r}$, $s\in\SP, r\in\SP^{\rm c}$,
be the admissible cut-offs defined in \eqref{eq:philrs}, 
and let  $\om^{(l)}_{s, r}$ and 
 $\varkappa^{(l)}_{s, r}$ be their canonical factors detailed in 
 \eqref{eq:factor1r} and \eqref{eq:factorsr} respectively. 
Using the factor $\tau$ from the canonical factorization  
$\mu(\hat\bx) = \tau(x_1, \hat\bx)\varkappa(\hat\bx)$  we define a new admissible cut-off 
\begin{align*}
\mu^{(l)}_{s, r}(x_1, \hat\bx) = \tau(x_1, \hat\bx) 
\varkappa^{(l)}_{s, r}(\hat\bx).
\end{align*}
It is clear that $\phi^{(l)}_{s, r}$ and $\mu^{(l)}_{s, r}$ 
are coupled to each other. 
Introduce the associated extended cut-off:
\begin{align}\label{eq:philmu}
\Phi^{(l)}_{s, r} (x, y, \hat\bx) =  \phi^{(l)}_{s, r}(x, \hat\bx) 
\tau(y, \hat\bx)
= \om^{(l)}_{s, r}(x, \hat\bx) \mu^{(l)}_{s, r}(y, \hat\bx). 
\end{align}
It follows from \eqref{eq:philrs}, \eqref{eq:factorsr} 
and \eqref{eq:supp}, \eqref{eq:suppderiv} 
that  
\begin{align}\label{eq:supports}
\supp\phi^{(l)}_{s, r}\subset \supp\phi, \quad 
\supp\mu^{(l)}_{s, r}\subset \supp\mu. 
\end{align}
Now we can describe the cluster derivatives of $\Phi$. The notation 
$\SD_{\SP}^l\Phi(x, y, \hat\bx)$ means taking the $l$th $\SP$-cluster derivative w.r.t. 
the variable $\bx = (x, \hat\bx)$.

\begin{lem}\label{lem:derivext}
Let $\Phi = \Phi(x, y, \hat\bx)$ be an extended cutoff 
associated with the coupled admissible $\phi$  and $\mu$ 
and assume that $\SP^*\subset \SfS^{\rm c}$. Let $l\in\mathbb N_0^3$ 
be an arbitrary multi-index such that $|l|=1$. 

If $\SP^{\rm c} = \varnothing$, i.e. $\SP = \SR$, then $\SD_\SP^l\Phi(x, y, \hat\bx) = 0$ 
for all $(x, y)\in \CD_\varepsilon$. 

If $\SP^{\rm c}\not = \varnothing$, then  
\begin{align}\label{eq:derivext}
\SD_{\SP}^l \Phi(x, y, \hat\bx) = \sum_{s\in\SP, r\in \SP^{\rm c}} 
\Phi^{(l)}_{s, r} (x, y, \hat\bx),
\quad \textup{for all}\ (x, y)\in \CD_\varepsilon,
\end{align}
where the extended cut-offs $\Phi^{(l)}_{s, r} (x, y, \hat\bx)$ 
are defined in \eqref{eq:philmu}. 
\end{lem}

\begin{proof}
According to \eqref{eq:cut}, 
\begin{align*}
\Phi(x, y, \hat\bx) = 
\phi(x, \hat\bx) 
\mu\big(y, \hat\bx; \{1\}, \SP^{\rm c}\big),\quad (x, y) \in\CD_\varepsilon. 
\end{align*}
Therefore, if $\SP^{\rm c} = \varnothing$, then 
$\Phi(x, y, \hat\bx) = \phi(x, \hat\bx)$. As $\SD_{\SP}^l\phi(\bx) = 0$, we also have 
  $\SD_\SP^l\Phi(x, y, \hat\bx) = 0$, as claimed. 

Assume that $\SP^{\rm c}\not = \varnothing$. By \eqref{eq:supports} and \eqref{eq:cut}, 
\begin{align}
\Phi^{(l)}_{s, r}(x, y, \hat\bx) = 
\phi^{(l)}_{s, r}(x, \hat\bx) 
\mu\big(y, \hat\bx; \{1\}, \SP^{\rm c}\big),
\quad s\in\SP, r\in\SP^{\rm c},
\label{eq:restr1}
\end{align}
for all $(x, y)\in \CD_\varepsilon$. 
Since the factor $\mu\big(y, \hat\bx; \{1\}, \SP^{\rm c}\big)$ 
does not depend on $x_j$ with $j\in\SP^*$, we have 
\begin{align*}
\SD_{\SP}^l \Phi(x, y, \hat\bx) = &\ D_{\SP}^l\phi(x, \hat\bx) 
\, \mu\big(y, \hat\bx; \{1\}, \SP^{\rm c}\big) + \phi(x, \hat\bx) \SD_{\SP^*}^l \mu\big(y, \hat\bx; \{1\}, \SP^{\rm c}\big)\\
= &\ D_{\SP}^l\phi(x, \hat\bx)\,  
\mu\big(y, \hat\bx; \{1\}, \SP^{\rm c}\big).
\end{align*}
Using \eqref{eq:derivrep} and \eqref{eq:restr1} we get 
\begin{align*}
\SD_{\SP}^l \Phi(x, y, \hat\bx) = \sum_{s\in\SP, r\in \SP^{\rm c}} \phi^{(l)}_{s, r}(x, \hat\bx)\,  
\mu\big(y, \hat\bx; \{1\}, \SP^{\rm c}\big) 
= \sum_{s\in\SP, r\in \SP^{\rm c}} \Phi^{(l)}_{s, r}(x, y, \hat\bx),
\end{align*}
as required.
\end{proof}
   
\section{Estimating the density matrix}\label{sect:estimates}

The proof of Theorem \ref{thm:deriv2} which is given in Sect. \ref{sect:proof}, 
uses a partition of unity that consists of extended cut-off functions, i.e. functions 
of the form \eqref{eq:Phi}. Thus the objective of this section is to estimate the derivatives of the function 
\begin{align}\label{eq:dens0}
\g(x, y; \Phi)  
= \int_{\R^{3(N-1)}} 
\psi(x, \hat\bx)\overline{\psi(y, \hat\bx)} \Phi(x, y, \hat\bx) 
d\hat\bx, 
\end{align}
with an extended cut-off $\Phi$ on the set $\CD_\varepsilon$:

\begin{thm}\label{thm:deriv}
Let $\Phi$ be an extended cut-off. Then for all $m, k\in\mathbb N_0^{3}$ we have 
\begin{align*}
\|\p_{x}^{k}\p_y^{m}\g 
(\ \cdot\ , \ \cdot\ ; \Phi)\|_{\plainL2(\CD_\varepsilon)}\le A^{|k|+|m|+2} 
(|k|+|m|+1)^{|k|+|m|}.
\end{align*}
The constant $A$ depends on $\varepsilon>0$, but does not depend on $\Phi$. 
\end{thm}

The proof of Theorem \ref{thm:deriv} is conducted by induction and 
hence we have to involve a more general object than \eqref{eq:dens0}. Precisely, 
for cluster sets $\BP=\{\SP_1, \SP_2, \dots, \SP_M\}$, $\BS=\{\SfS_1, \SfS_2, \dots, \SfS_K\}$, 
and multi-indices 
$\bk\in \mathbb N_0^{3M}, \bm\in \mathbb N_0^{3K}$,  
introduce the function  
\begin{align}\label{eq:dens}
\g_{\bk, \bm}(x, y; \BP, &\ \BS; \Phi) 
= \int_{\R^{3(N-1)}} 
\SD_{\BP}^{\bk}\psi(x, \hat\bx)\overline{\SD_{\BS}^{\bm}\psi(y, \hat\bx)} 
\Phi(x, y, \hat\bx) 
d\hat\bx.
\end{align}
If $\bm = \mathbf 0$ (and/or $\bk = \mathbf 0$), then 
this integral is independent 
of $\BP$ (and/or $\BS$), and 
in this case we set $\BP = \varnothing$ (and/or $\BS = \varnothing$).  
Thus $\g_{\mathbf 0, \mathbf 0}(x, y; \varnothing, \varnothing; \Phi)$ coincides with the function 
in \eqref{eq:dens0}.
Note the symmetry of $\g_{\bk, \bm}$:
\begin{align}\label{eq:sym}
\g_{\bk, \bm}(x, y; \BP, \BS; \Phi)  = \overline{\g_{\bm, \bk}(y, x; \BS, \BP; \tilde\Phi) },\quad 
\tilde\Phi(y, x; \hat\bx) = \Phi(x, y; \hat\bx).
\end{align}
We estimate the function $\g_{\bk, \bm}$ and its derivatives 
on the set $\CD_\varepsilon, \varepsilon>0$, 
defined in \eqref{eq:deps} with the help of Corollary \ref{cor:derivl2} 
by reducing the estimates to the integrals $\|\SD_\BP^{\bk}\psi\|_{\plainL2(U_{\BP}(\varepsilon))}$ 
and $\|\SD_\BS^{\bm}\psi\|_{\plainL2(U_\BS(\varepsilon))}$.  
To this end we assume that the 
admissible cut-offs  $\phi$ and $\mu$ associated with $\Phi$  
satisfy the following support conditions:
\begin{align}\label{eq:suppembed1}
\supp\phi\subset X_{\BP}\big(\varepsilon(4N)^{-1}\big), \quad 
\supp\mu\subset X_{\BS}\big(\varepsilon(4N)^{-1}\big),
\end{align}
and   
\begin{align}\label{eq:suppembed2}
\supp\phi(x,\ \cdot\ ) \cap \supp\mu(y,\ \cdot\ )\subset 
\widehat T_{\BP^*}(\varepsilon/2)\cap \widehat T_{\BS^*}(\varepsilon/2),
\quad 
\textup{for all}\  (x, y)\subset \CD_\varepsilon. 
\end{align}
Note that 
conditions \eqref{eq:suppembed1} and \eqref{eq:suppembed2} are automatically satified 
for $\BP = \{\SP\}$, 
$\BS = \{\SfS\}$, where as usual $\SP = \SQ(\phi)$, $\SfS = \SQ(\mu)$.
Indeed, \eqref{eq:suppembed1} follows from \eqref{eq:suppphi}, 
and \eqref{eq:suppembed2} follows from \eqref{eq:phihat}.

Recall that the sets $X, T, \widehat T$ with various subscripts 
are defined in \eqref{eq:xp}, \eqref{eq:tp}, \eqref{eq:hattp}, \eqref{eq:up}. 
For brevity, throughout the proofs below for an arbitrary cluster set $\BQ$  we 
use the notation $T_\BQ = T_\BQ(\varepsilon/2)$, 
$\widehat T_\BQ = \widehat T_\BQ(\varepsilon/2)$ 
and $X_\BQ = X_\BQ(\varepsilon(4N)^{-1})$.

\begin{lem}\label{lem:base}
Suppose that $\Phi$ is of the form \eqref{eq:Phi} and  
that \eqref{eq:suppembed1} 
and \eqref{eq:suppembed2} hold. Then there exists a constant $A_3$, independent of 
the cluster sets $\BP, \BS$, and of the cut-off $\Phi$, such that 
\begin{align*}
\|\gamma_{\bk, \bm}(\ \cdot\ , \ \cdot\ ; \BP,  \BS; \Phi)\|_{\plainL2(\CD_\varepsilon)}
\le A_3^{|\bk|+|\bm|+2} (|\bk|+|\bm|+1)^{|\bk|+|\bm|},
\end{align*}
for all $\bk\in \mathbb N_0^{3M}$, $\bm\in\mathbb N_0^{3K}$.  
\end{lem}

\begin{proof}   
Let 
\begin{align*}
C_a = \max \{1, \max_{l: |l|=1}\|\p^l\t\|^{N^2}\},
\end{align*}
so that $|\om|, |\tau|, |\varkappa|, |\phi|, |\mu|\le C_a$. 
Therefore
\begin{align*}
|\Phi(x, y, \hat\bx)| = &\ |\om(x, \hat\bx)||\tau(y, \hat\bx)||\varkappa(\hat\bx)|\\
= &\ |\om(x, \hat\bx)|^{\frac{1}{2}}|\tau(y, \hat\bx)|^{\frac{1}{2}}
|\om(x, \hat\bx)\varkappa(\hat\bx)|^{\frac{1}{2}} 
 |\tau(y, \hat\bx)\varkappa(\hat\bx)|^{\frac{1}{2}} 
 \le C_a |\phi(x, \hat\bx)|^{\frac{1}{2}}|\mu(y, \hat\bx)|^{\frac{1}{2}}.
\end{align*}                                                                                                                                                                                                                                                                                                                                                                                                                                                                                                                                                                                                                                                                                                                                                                                                                                                                                                                                                                                                                                                                                                                                                                                                                                                                                                                                                                                                                                                                                                                                                                                                                                                                                                                                                                                                                                                                                                                                                                                                                                                                                                                                                                                                                                                                                                                                                                                                                                                                                                                                                                                                                                                                                                                                                                     
Now, using \eqref{eq:suppembed2}, we can estimate:
\begin{align*}
\|\gamma_{\bk, \bm}(\ \cdot\ , \ \cdot\ ; &\ \BP,  \BS; \Phi)\|_{\plainL2(\CD)}^2\\
\le &\ C_a^2 \int\limits_{|x|>\varepsilon}
\int\limits_{|y|>\varepsilon}   
\bigg[\int\limits_{\widehat T_{\BP^*}\cap \widehat T_{\BS^*}} 
|\SD_{\BP}^{\bk}\psi(x, \hat\bx)| |\SD_{\BS}^{\bm}\psi(y, \hat\bx)| 
|\phi(x, \hat\bx)|^{\frac{1}{2}}
|\mu(y, \hat\bx)|^{\frac{1}{2}}
d\hat\bx\bigg]^2\, dy dx.
\end{align*}
By  H\"older's inequality and by \eqref{eq:suppembed1}, 
the right-hand side does not exceed
\begin{align*}
 &\  C_a^2\int\limits_{|x|> \varepsilon} \ \int\limits_{\widehat T_{\BP^*}} 
|\SD_{\BP}^{\bk}\psi(x, \hat\bx)|^2 |\phi(x, \hat\bx)|d\hat\bx \, dx
\int\limits_{|y|> \varepsilon} \ \int\limits_{\widehat T_{\BS^*}} 
|\SD_{\BS}^{\bm}\psi(y, \hat\bx)|^2 |\mu(y, \hat\bx)| d\hat\bx \, dy\\[0.2cm]
&\ \qquad\le 
C_a^4 \iint\limits_{X_{\BP}\cap T_{\BP}} 
|\SD_{\BP}^{\bk}\psi(x, \hat\bx)|^2 \ d\hat\bx\, dx 
\iint\limits_{X_{\BS}\cap T_{\BS}} |\SD_{\BS}^{\bm}\psi(y, \hat\bx)|^2\ d\hat\bx\, dy.
\end{align*}
Since $X_{\BP}\cap T_{\BP}\subset U_{\BP}(\varepsilon(4N)^{-1})$
(see the definition \eqref{eq:up}), and a similar inclusion holds for 
the cluster set $\BS$, by Corollary \ref{cor:derivl2}, the right-hand side does not exceed
\begin{align*}
C_a^4 A_2^{2(|\bk|+|\bm|+2)}(|\bk|+1)^{2|\bk|}(|\bm|+1)^{2|\bm|}
\le A_3^{2(|\bk|+|\bm|+2)}(|\bk|+|\bm|+1)^{2(|\bk|+|\bm|)},
\end{align*}
with $A_3 = C_a^2A_2$. This implies the required bound. 
\end{proof}

Let the functions $\phi^{(l)}_{s, r}$ and $\mu^{(l)}_{s, r}, \Phi^{(l)}_{s, r}$ 
be as defined in \eqref{eq:philrs} and 
\eqref{eq:philmu}  
respectively.  
As in the previous section, 
we use the notation $\SP = \SQ(\phi)$ and $\SfS = \SQ(\mu)$. In the next lemma we show how the derivatives 
of $\gamma_{\bk, \bm}$ w.r.t. the variable $x$ transform into directional derivatives under the integral 
\eqref{eq:dens}.

\begin{lem} \label{lem:oneup}
Suppose that \eqref{eq:suppembed1} and \eqref{eq:suppembed2} hold.  
Assume that $\SP^*\subset\SfS^{\rm c}$. Then 
\begin{align}\label{eq:phimu}
\supp \phi\subset X_{\{\SP, \BP\}},\quad \supp\mu\subset X_{\{\SP^{*}, \BS\}}.
\end{align}
Furthermore,
\begin{align}\label{eq:phixy}
\supp\phi(x,\ \cdot\ )\cap \supp\mu(y,\ \cdot\ )\subset 
 \widehat T_{\{\SP^*, \BP^*\}} \cap \widehat T_{\{\SP^*, \BS^*\}},
 \quad \textup{for all}\  (x, y)\in\CD_\varepsilon. 
\end{align}

Let $l\in \mathbb N_0^3$ be such that $|l|=1$.  
If $\SP^{\rm c} = \varnothing$, 
then 
\begin{align}\label{eq:stepupp}
\p_x^l\g_{\bk, \bm}(x, y; \BP, \BS; \Phi) = 
\g_{(l, \bk), \bm}\big(x, y; \{\SP, \BP\}, \BS; \Phi\big) 
+ \g_{\bk, (l, \bm)}\big(x, y; \BP, \{\SP^*, \BS\}; \Phi\big),
\end{align} 
and both sides are square-integrable in $(x, y)\in \CD_\varepsilon$.

If $\SP^{\rm c}\not = \varnothing$, 
then for all $s\in \SP, r\in \SP^{\rm c}$ we have 
\begin{align}\label{eq:phimul}
\supp\phi^{(l)}_{s, r}\subset X_{\BP}, \quad \supp\mu^{(l)}_{s, r}\subset X_{\BS},
\end{align}
and 
\begin{align}\label{eq:phixyl}
\supp\phi^{(l)}_{s, r}(x,\ \cdot\ )\cap 
\supp\mu^{(l)}_{s, r}(y,\ \cdot\ )
\subset \widehat T_{\BP^*}\cap \widehat T_{\BS^*},
\end{align}
for every $(x, y)\in \CD_\varepsilon$. 
Furthermore, the formula holds:
\begin{align}\label{eq:stepup}
\p_x^l\g_{\bk, \bm}(x, y; \BP, \BS; \Phi) = &\ 
\g_{(l, \bk), \bm}\big(x, y; \{\SP, \BP\}, \BS; \Phi\big) 
+ \g_{\bk, (l, \bm)}\big(x, y; \BP, \{\SP^*, \BS\}; \Phi\big)\notag\\
&\ + 
\sum_{s\in \SP, r\in \SP^{\rm c}}\g_{\bk, \bm}\big(x, y; \BP, \BS; \Phi^{(l)}_{s, r}\big),
\end{align} 
and both sides are square-integrable in $(x, y)\in \CD_\varepsilon$.
\end{lem}

\begin{proof} 
According to \eqref{eq:phimux} and the assumption \eqref{eq:suppembed1}, we have 
\begin{align*}
\supp\phi\subset X_{\SP}\cap X_{\BP} = X_{\{\SP, \BP\}},\quad 
\supp\mu\subset X_{\SP^*}\cap X_{\BS} = X_{\{\SP^*, \BS\}}, 
\end{align*}
which coincides with \eqref{eq:phimu}. Moreover, by \eqref{eq:phihat} and \eqref{eq:suppembed2},  
\begin{align*}
\supp\phi(x,\ \cdot\ )\cap\supp\mu(y,\ \cdot\ )
\subset \widehat T_{\SP^*} \cap \widehat T_{\BP^*}\cap \widehat T_{\BS^*},
\end{align*}
which implies \eqref{eq:phixy} for all $(x, y)\in \CD_\varepsilon$. 
Thus by Lemma \ref{lem:base} the terms on the right-hand side of 
\eqref{eq:stepupp} and two first terms on the right-hand side of \eqref{eq:stepup} 
are square-integrable in $(x, y)\in \CD_\varepsilon$. 

Let $\SP^{\rm c}\not = \varnothing$. Then the inclusions 
\eqref{eq:phimul} and \eqref{eq:phixyl} are consequences of \eqref{eq:supports} 
and \eqref{eq:suppembed1}, \eqref{eq:suppembed2}. 
Again by Lemma \ref{lem:base}, the third term on the right-hand side of \eqref{eq:stepup} 
is square-integrable 
in $(x, y)\in \CD_\varepsilon$, as claimed. 

It remains to prove \eqref{eq:stepupp} and \eqref{eq:stepup}. 
We assume throughout that $(x, y)\in\CD_\varepsilon$. 
Introduce the vector $\hat\bz = (z_2, z_3, \dots, z_N)$ such that 
 \begin{align*}
 z_j = 
 \begin{cases}
 x, \quad j\in \SP^*,\\
 0, \quad j\in \SP^{\rm c},
 \end{cases}
\end{align*}
and make the following change of variables under the integral
\eqref{eq:dens}: $\hat\bx = \hat\bw + \hat\bz$, so  
\begin{align*}
\g_{\bk, \bm}(x, y; \BP, \BS; \Phi ) = \int_{\R^{3(N-1)}} 
\SD_{\BP}^{\bk} \psi(x, \hat\bw+\hat\bz) \overline{\SD_{\BS}^{\bm}\psi(y, \hat\bw+\hat\bz)}
\Phi(x, y, \hat\bw+\hat\bz) d\hat\bw.
\end{align*} 
Before differentiating this integral w.r.t. $x$, we make the following useful observation. 
For each  $l\in \mathbb N_0^3$ 
and for arbitrary functions $g=g(\bx)$ and 
$h = h(x, y, \hat\bx)$, due to the definition of $\hat\bz = \hat\bz(x)$, we 
have 
\begin{align*}
\p_x^l\big( g(x, \hat\bw+\hat\bz)\big) = &\ \big(\SD_\SP^l g\big)(x, \hat\bw+\hat\bz),\quad 
\p_x^l\big( g(y, \hat\bw+\hat\bz)\big)
= \big(\SD_{\SP^*}^l g\big)(y, \hat\bw+\hat\bz)
\\[0.2cm]
&\ \quad\textup{and}\quad \p_x^l\big( h(x, y, \hat\bw+\hat\bz)\big) 
= \big(\SD_{\SP}^l h\big)(x, y, \hat\bw+\hat\bz).
\end{align*}
Therefore, for $|l|=1$ we have 
\begin{align*}
\p_x^l\g_{\bk, \bm}(x, y; \BP, \BS; \Phi) = &\ \int_{\R^{3(N-1)}} 
\SD_{\SP}^l\SD_{\BP}^{\bk} \psi(x, \hat\bw+\hat\bz) 
\overline{\SD_{\BS}^{\bm}\psi(y, \hat\bw+\hat\bz)} 
\Phi(x, y, \hat\bw+\hat\bz) d\hat\bw\notag\\
 + &\ \int_{\R^{3(N-1)}} 
\SD_{\BP}^{\bk} \psi(x, \hat\bw+\hat\bz) 
\overline{\SD_{\SP^*}^l\SD_{\BS}^{\bm}\psi(y, \hat\bw+\hat\bz)}
\Phi(x, y, \hat\bw+\hat\bz) d\hat\bw\notag\\
+ &\ \int_{\R^{3(N-1)}} 
\SD_{\BP}^{\bk} \psi(x, \hat\bw+\hat\bz) 
\overline{\SD_{\BS}^{\bm}\psi(y, \hat\bw+\hat\bz)}\,
\SD_\SP^l\Phi(x, y, \hat\bw+\hat\bz)d\hat\bw.
\end{align*} 
Changing the variables back to $\hat\bx$, we rewrite the right-hand side as 
\begin{align}\label{eq:ghaty}
 &\ \int_{\R^{3(N-1)}} 
\SD_{\SP}^l\SD_{\BP}^{\bk} \psi(x, \hat\bx) \overline{\SD_{\BS}^{\bm}\psi(y, \hat\bx)} 
\,\Phi(x, y, \hat\bx) d\hat\bx\notag\\
 &\ \qquad\qquad + \int_{\R^{3(N-1)}} 
\SD_{\BP}^{\bk} \psi(x, \hat\bx) \overline{\SD_{ \SP^*}^l\SD_{\BS}^{\bm}\psi(y, \hat\bx)}
\,\Phi(x, y, \hat\bx) d\hat\bx\notag\\
&\ \qquad\qquad +  \int_{\R^{3(N-1)}} 
\SD_{\BP}^{\bk} \psi(x, \hat\bx) 
\overline{\SD_{\BS}^{\bm}\psi(y, \hat\bx)}
\, \SD_\SP^l\Phi(x, y, \hat\bx) d\hat\bx.
\end{align}
In the case $\SP^{\rm c} = \varnothing$, we have $\SD_\SP^l\Phi(x, y, \hat\bx) = 0$ by Lemma 
\ref{lem:derivext}, so the sum coincides with the right-hand side of \eqref{eq:stepupp}.  

If $\SP^{\rm c}\not = \varnothing$, then by \eqref{eq:derivext} the sum 
\eqref{eq:ghaty}
coincides with 
the right-hand side of \eqref{eq:stepup} again. This completes the proof. 
\end{proof} 
 
\begin{prop}\label{prop:deriv} 
Let $\Phi$ be an extended cut-off, and let $\BP, \BS$ be a pair of cluster sets 
such that \eqref{eq:suppembed1} and \eqref{eq:suppembed2} hold. 
 Then for all $m, k\in\mathbb N_0^3$ we have 
\begin{align}\label{eq:deriv}
\|\p_{x}^{k}\p_y^{m}\gamma_{\bk, \bm}
(\ \cdot\ &, \ \cdot\ ; \BP, \BS; \Phi)\|_{\plainL2(\CD_\varepsilon)}\notag\\
\le &\ A^{|k|+|m|} A_3^{|\bk|+|\bm|+2} 
(|\bk|+|\bm|+|k|+|m|+1)^{|\bk|+|\bm| + |k|+|m|},
\end{align}
where the constant $A_3$ is as in Lemma \ref{lem:base} and $A = 2A_3 + N^2$. 
\end{prop}

The proof of this proposition is by induction. Lemma \ref{lem:base} provides the base step. 
Now we need to establish the induction step:

\begin{lem}\label{lem:indstep} 
Suppose that for every extended cut-off $\Phi$ and every pair $\BP=\{\SP_1, \SP_2, \dots, \SP_M\}$ and 
$\BS=\{\SfS_1, \SfS_2, \dots, \SfS_K\}$ of cluster sets such that 
$\Phi$ satisfies \eqref{eq:suppembed1} and \eqref{eq:suppembed2}, the bound \eqref{eq:deriv} holds 
for all 
multi-indices $\bk\in\mathbb N_0^{3M}, \bm\in \mathbb N_0^{3K}$, and all 
$k, m\in\mathbb N_0^3$, such that $|k|\le p$, $|m|\le n$ with 
some $p, n\in \mathbb N_0$. 
Then for such extended cut-offs $\Phi$ and cluster sets $\BP, \BS$ 
the bound \eqref{eq:deriv} holds  for all 
$k, m$, such that $|k|\le p+1$, $|m|\le n$.
\end{lem}

\begin{proof}
 Let $|m|\le n$, $k=k_0+l$ with $l\in\mathbb N_0^3$, $|l|=1$, and  arbitrary $k_0\in\mathbb N_0^{3}$, such that 
$|k_0|=p$. 
In view of Lemma \ref{lem:empty}, we may assume that $\SP^*\subset\SfS^{\rm c}$, since 
otherwise the integrand in \eqref{eq:dens} equals zero. 
Thus we can apply Lemma \ref{lem:oneup}. Assume first that $\SP^{\rm c}\not = \varnothing$. 
It follows from \eqref{eq:stepup} that 
\begin{align}\label{eq:oneup}
\p_x^{k_0+l}\p_y^{m}\g_{\bk, \bm}&\ (x, y; \BP, \BS; \Phi)\notag\\ 
= &\ \p_x^{k_0}\p_y^{m}\g_{(l, \bk), \bm}\big(x, y; \{\SP, \BP\}, \BS; \Phi\big) 
+ \p_x^{k_0}\p_y^{m}\g_{\bk, (l, \bm)}\big(x, y; \BP, \{\SP^*, \BS\}; \Phi\big)\notag\\
&\ + \sum_{s\in \SP, r\in \SP^{\rm c}}\p_x^{k_0}
\p_y^{m}\g_{\bk, \bm}\big(x, y; \BP, \BS; \Phi^{(l)}_{s, r}\big).
\end{align}
According to \eqref{eq:phimu}, \eqref{eq:phixy}, 
the function $\Phi$ satisfies the 
conditions \eqref{eq:suppembed1} and \eqref{eq:suppembed2} for the pair of cluster sets 
$\{\SP, \BP\}, \BS$ and also for the pair $\BP, \{\SP^*, \BS\}$.   
Similarly, due to \eqref{eq:phimul} and \eqref{eq:phixyl}, 
the extended cut-off 
function $\Phi^{(l)}_{s, r}$ satisfies \eqref{eq:suppembed1}, \eqref{eq:suppembed2} 
for the pair of cluster sets $\BP, \BS$.
Since $|m|\le n$ and $|k_0|=p$, by
the assumptions of the lemma, each term on the right-hand side 
of \eqref{eq:oneup} satisfies the bound of the form \eqref{eq:deriv}. With the notation 
$q = |m|$ this gives the following estimate:
\begin{align*}
\|\p_x^{k_0+l}\p_y^{m}
\g_{\bk, \bm}(\ \cdot\ ,\ \cdot\ ; &\ \BP, \BS; \Phi)\|_{\plainL2(\CD_\varepsilon)}\\
\le  &\ 
2 A^{p+q}A_3^{|\bk|+|\bm|+3} (|\bk|+|\bm|+p+q+2)^{|\bk|+|\bm|+p+q +1} 
\notag\\
&\ + N^2 A^{p+q}A_3^{|\bk|+ |\bm|+2} (|\bk|+|\bm|+p+q+1)^{|\bk|+|\bm|+p+q} 
\\
\le &\ A^{p+q} A_3^{|\bk|+ |\bm|+2}
\big( 2 A_3 + N^2 \big)
(|\bk|+|\bm|+p+q+2)^{|\bk|+|\bm|+p+q +1}. 
\end{align*}
Setting $A = 2A_3+ N^2$, we get \eqref{eq:deriv} with $k=k_0+l$, as required. 

If $\SP^{\rm c} = \varnothing$, then the only difference in the proof is that 
instead of \eqref{eq:stepup} we use \eqref{eq:stepupp}. 
\end{proof}

\begin{proof}[Proof of Proposition \ref{prop:deriv}]

\underline{Step 1. Proof of \eqref{eq:deriv} for all $k$ and $m = 0$.} 
According to Lemma \ref{lem:base}, the required bound holds for $k=m=0$. 
Thus, using Lemma \ref{lem:indstep}, by induction we conclude that 
\eqref{eq:deriv} holds for all $k\in\mathbb N_0^3$ and $m=0$, 
as claimed.

\underline{Step 2. Proof of \eqref{eq:deriv} for $k=0$ and all $m$.} 
Using the symmetry property \eqref{eq:sym} and Step 1, we conclude that 
\eqref{eq:deriv} holds for all $m\in\mathbb N_0^3$ and $k = 0$.

\underline{Step 3.}
Using Step 2 and Lemma \ref{lem:indstep}, 
by induction we conclude that 
\eqref{eq:deriv} holds for all $k, m\in\mathbb N_0^3$, 
as required.
\end{proof}
 
\begin{proof}[Proof of Theorem \ref{thm:deriv}]
Recall that in the case $\bm = \mathbf 0, \bk = \mathbf 0$ 
we take $\BP=\BS = \varnothing$, so that the conditions 
\eqref{eq:suppembed1} and \eqref{eq:suppembed2} are automatically satisfied. Thus the required bound 
follows directly from \eqref{eq:deriv}.
\end{proof} 

\section{Proof of Theorem \ref{thm:deriv2}}\label{sect:proof}

First we build a suitable partition of unity, using the functions $\z$ and $\t$, defined in 
\eqref{eq:pu}. Recall the notation $\SR = \{1, 2, \dots, N\}$. 

Let $\Xi = \{(j, k)\in\SR\times\SR: j < k\}$. For each subset $\U\subset\Xi$ 
introduce the admissible cut-off 
\begin{align*}
\phi_{\U}(\bx) = \prod_{(j, k)\in \U} \z(x_j-x_k) \prod_{(j, k)\in\U^{\rm c}} \t(x_j-x_k).
\end{align*}
It is clear that 
\begin{align*}
\sum_{\U\subset{\Xi}} \phi_{\Upsilon}(\bx) 
= \prod_{(j, k)\in \Xi}\big(\z(x_j-x_k)+ \t(x_j-x_k)\big) 
= 1.
\end{align*}
For every cluster $\SfS\subset \SR^*$ define 
\begin{align*}
\tau_{\SfS}(x_1, \hat\bx) = \prod_{j\in \SfS} \z(x_1-x_j) \prod_{j\in(\SfS^{\rm c})^*} \t(x_1-x_j).
\end{align*}
It is clear that 
\begin{align*}
\sum_{\SfS\subset \SR^*}\tau_{\SfS}(x_1, \hat\bx) 
= \prod_{j\in\SR^*}\big(\z(x_1-x_j)+ \t(x_1-x_j)\big) = 1.  
\end{align*}
Define
\begin{align*}
\Phi_{\U, \SfS}(x, y, \hat\bx) = \phi_{\U}(x, \hat\bx) \tau_{\SfS}(y, \hat\bx),\quad 
(x, y)\in \R^3\times \R^3, \hat\bx\in \R^{3N-3},
\end{align*}
so that 
\begin{align*}
\sum\limits_{\U\subset\Xi,\ \SfS\subset\SR^*}\Phi_{\U, \SfS}(x, y, \hat\bx) = 1.
\end{align*}
 Note that each function $\Phi_{\U, \SfS}$ is an extended cut-off function, as 
defined in Subsect. \ref{subsect:extended}. 
Using the definition \eqref{eq:dens0}, 
the function \eqref{eq:den} can be represented as 
 \begin{align*}
 \g(x, y) = \sum_{\U\subset\Xi, \ \SfS\subset\SR^*} \g(x, y; \Phi_{\U, \SfS}).
 \end{align*}
 Since each function $\Phi_{\U, \SfS}$ is an extended cut-off, 
we can use Theorem \ref{thm:deriv} for each term, which 
 leads to \eqref{eq:der}, as required. 
 
 As explained in Sect. \ref{sect:main}, Theorem \ref{thm:deriv2} implies Theorem \ref{thm:main1}, 
 and hence Theorem \ref{thm:main}.

\section{Appendix: elementary combinatorial formulas}

Here we collect some elementary formulas.  

\subsection{Stirling's formula} 
It follows from Stirling's formula 
\begin{align*}
\lim_{p\to\infty} \frac{p! e^p}{p^{p+\frac{1}{2}}} = \sqrt{2\pi}
\end{align*}
that 
\begin{align}\label{eq:comb11}
C^{-1}(p+1)^{p+\frac{1}{2}} e^{-p} \le p! \le C (p+1)^{p+\frac{1}{2}} e^{-p},
\end{align}
for all $p = 0, 1, 2, \dots$. 
Therefore
\begin{align}\label{eq:comb1}
(p+1)^p\le Ce^p \ p!,\quad \forall p\in \mathbb N_0.
\end{align}
The bounds \eqref{eq:comb11} also imply 
for any $p = 0, 1, \dots$ and $q = 0, 1, \dots, p,$ that 
\begin{align}\label{eq:choosepow1}
{p\choose q} = \frac{p!}{q!(p-q)!}
\le L_3\frac{(p+1)^p}{(q+1)^q (p-q+1)^{p-q}},
\end{align}
with some constant $L_3>0$, independent of $p$ and $q$.

\subsection{Multiindices and factorials}
For $k = (k_1, k_2, \dots, k_d)\in \mathbb N_0^d$, we use the standard notation
\begin{align*}
k! = k_1! k_2! \cdots k_d!,\ \quad |k| = |k_1|+|k_2|+\dots+|k_d|.
\end{align*}
We say that $k\le s$ for $k, s\in\mathbb N_0^d$ if $k_j\le s_j$, $j = 1, 2, \dots, d$. 
In this case we define 
\begin{align*}
{k \choose s} = \frac{k!}{s!(k-s)!}.
\end{align*}
Note the useful identity 
\begin{align}\label{eq:choose}
\sum_{\substack{l\le k\\|l| = p}}{k\choose l} = {|k|\choose p},\ \forall p \le |k|.
\end{align}
It follows by comparing the coefficients of the term $t^p$ in the 
expansions of both sides of the 
equality
\begin{align*}
(1+t)^{k_1} (1+t)^{k_2}\cdots (1+t)^{k_d} = (1+t)^{|k|},\quad t\in\R.
\end{align*}
This simple argument is found in \cite[Proposition 2.1]{KKato}.

And to conclude, the multinomial formula  (see e.g. \cite[\S 24.1.2]{AbrSteg})
\begin{align*}
d^p = \bigg(\sum_{l=1}^d 1\bigg)^p = \sum\limits_{\substack{k\in\mathbb N_0^d\\ |k|=p}}\frac{|k|!}{k!}
\end{align*}
implies that 
\begin{align}\label{eq:spr}
|k|!\le d^{|k|} k!,\quad \forall k\in \mathbb N_0^d. 
\end{align} 

\vskip 0.5cm

\textbf{Acknowledgments.} The authors are grateful to S. Fournais, T. Hoffman-Ostenhof, 
M. Lewin and T. \O. S\o rensen
for stumulating discussions and advice.
The second author was supported by the EPSRC grant EP/P024793/1.

\end{document}